\documentclass[12pt,reqno]{amsart}
\usepackage{amsmath}
\usepackage{amssymb}
\usepackage{amsthm}
\usepackage{verbatim}
\usepackage{subfigure}
\usepackage{latexsym}
\usepackage{lmodern}
\usepackage{newalg}
\usepackage{tabularx}
\usepackage{array}
\usepackage{graphicx}
\usepackage{float}
\usepackage{multirow}
\usepackage{listings}
\usepackage{bm}
\usepackage{url}
\usepackage[english]{babel}

\usepackage{hyperref}
\usepackage{algorithm} 
\usepackage{algorithmic}

\usepackage{multicol,balance}
\usepackage{times}

\newcommand{\mb}{\bm{b}}
\newcommand{\mc}{\bm{c}}
\newcommand{\hmc}{\breve{\bm{c}}}
\newcommand{\hmx}{\breve{\x}}

\newcommand{\n}{\bm{n}}
\newcommand{\x}{\bm{x}}
\newcommand{\y}{\bm{y}}

\newcommand{\cA}{\mathcal{A}}
\newcommand{\cB}{\mathcal{B}}
\newcommand{\cP}{\mathcal{P}_{\Omega}}
\newcommand{\cH}{\mathcal{H}}
\newcommand{\ws}{\mathcal{D}}
\newcommand{\C}{\mathbb{C}}
\newcommand{\N}{\mathbb{N}}
\newcommand{\R}{\mathbb{R}}
\newcommand{\Z}{\mathbb{Z}}
\newcommand{\gep}{\varepsilon}

\newcommand{\wh}{\widehat}
\newcommand{\ol}{\overline}
\newcommand{\pP}{\mathsf{P}}
\newcommand{\be}{\begin{equation}}
\newcommand{\ee}{\end{equation}}
\newcommand{\la}{\langle}
\newcommand{\ra}{\rangle}
\newcommand{\setsp}{\;:\;}     
\newcommand{\tp}{\mathsf{T}}  

\renewcommand{\le}{\leqslant}
\renewcommand{\ge}{\geqslant}

\newtheorem{lemma}{Lemma}

\newtheorem{theorem}[lemma]{Theorem}

\newcommand{\tpctf}{\operatorname{TP-\mathbb{C}TF}}
\newcommand{\ctf}{\operatorname{\mathbb{C}TF}}
\DeclareMathOperator{\sgn}{sgn}
\DeclareMathOperator{\diag}{diag}
\DeclareMathOperator{\range}{Ran}
\DeclareMathOperator{\tol}{tol}
\DeclareMathOperator{\error}{error}

\DeclareMathOperator{\tolerence}{tolerence}
\DeclareMathOperator{\midd}{mid}

\numberwithin{equation}{section}

\normalsize \makeatletter

\begin{document}
\title[Image Inpainting Using Complex Tight Framelets]{Image Inpainting Using Directional Tensor Product Complex Tight Framelets}

\thanks{Research was supported in part by NSERC Canada. Research of Yi Shen was also supported
by an PIMS postdoctoral fellowship.}

\author{Yi Shen}

\address{Department of Mathematics and Science, Zhejiang Sci--Tech University, Hangzhou, China 310028. Also associated with Department of Mathematical and Statistical Sciences,
University of Alberta, Edmonton, Alberta T6G 2G1, Canada, and Department of Mathematics and Statistics, University of Calgary, Calgary, Alberta T2N 1N4, Canada. \quad E-mail: {\tt sy1133@gmail.com}}

\author{Bin Han}

\address{Department of Mathematical and Statistical Sciences,
University of Alberta, Edmonton, Alberta T6G 2G1, Canada.
\quad E-mail: {\tt bhan@ualberta.ca}\quad
{\tt http://www.ualberta.ca/$\sim$bhan}}

\author{Elena Braverman}

\address{Department of Mathematics and Statistics, University of Calgary, 2500 University Drive N.W., Calgary, Alberta T2N 1N4, Canada.\quad
E-mail: {\tt maelena@math.ucalgary.ca}}

\makeatletter \@addtoreset{equation}{section} \makeatother

\begin{abstract}
Different from an orthonormal basis, a tight frame is an overcomplete and energy-preserving system with flexibility and redundancy.
Many image restoration methods employing tight frames have been developed and studied in the literature. Tight wavelet frames have been proven to be useful in many applications.
In this paper we are particularly interested in the image inpainting problem using directional complex tight wavelet frames.
Under the assumption that frame coefficients of images are sparse,
several iterative thresholding algorithms for the image inpainting problem have been proposed in the literature. The outputs of such iterative algorithms are closely linked to solutions of several convex minimization models using the balanced approach which simultaneously combines the $l_1$-regularization for sparsity of frame coefficients and the $l_2$-regularization for smoothness of the solution.
Due to the redundancy of a tight frame, elements of a tight frame could be highly correlated and therefore, their corresponding frame coefficients of an image are
expected to close to each other. This is called the grouping effect in statistics. In this paper, we establish the grouping effect property for frame-based convex minimization models using the balanced approach. This result on grouping effect partially explains the effectiveness of models using the balanced approach for several image restoration problems. Since real-world natural images usually have two layers consisting of cartoons and textures, methods using simultaneous cartoon and texture inpainting are popular in the literature by using two combined tight frames: one tight frame (often built from wavelets, curvelets or shearlets) provides sparse representations for cartoons, and the other tight frame (often built from discrete cosine transform) offers sparse approximation for textures. Inspired by recent development on directional tensor product complex tight framelets ($\tpctf$s) and their impressive performance for the image denoising problem, in this paper we propose an iterative thresholding algorithm using a single tight frame derived from $\tpctf$s for the image inpainting problem.
Experimental results show that our proposed algorithm can handle well both cartoons and textures simultaneously and performs comparably and often better than several well-known frame-based iterative thresholding algorithms for the image inpainting problem without noise. For the image inpainting problem with additive zero-mean i.i.d. Gaussian noise, our proposed algorithm using $\tpctf$s performs superior than other known state-of-the-art frame-based image inpainting algorithms.
\end{abstract}


 \maketitle
\pagenumbering{arabic}

\section{Introduction and Motivations}
\label{sec:introduction}

Image restoration problems can be generally formulated as the following linear inverse problem:
\be \label{inverse}
\y=\cA \x+\n,
\ee
where $\y\in \R^d$ is a given observed image and the matrix $\cA\in \R^{d\times d}$ is often given, $\n$ is a noise vector, and $\x\in \R^d$ is the unknown ``true'' clean image to be restored/recovered. Note that for convenience of discussion, an image is often regarded as a column vector by stacking the columns of the image one by one. For the particular choice $\cA=I$ where $I$ is the identity matrix, \eqref{inverse} is the image denoising problem.
If $\cA$ is a diagonal matrix with its diagonal entries being either $1$ or $0$, it corresponds to the image inpainting problem. Many methods have been proposed to solve the linear inverse problem in \eqref{inverse}. Among them, one popular approach is to employ sparse representations and convex minimization schemes with regularization (\cite{cai2008framelet,cai2010simultaneous,cai2009split,Chan2005,elad:book,elad2005simultaneous,li2013adaptive,lim2013nonseparable} and many references therein).
Roughly speaking, one expects that the unknown clean image $\x$ in \eqref{inverse} has a sparse representation under a basis, or a frame, or more generally a dictionary. Then one hopes to recover the unknown image $\x$ by finding the few dominating large coefficients of $\x$ in the transform domain from \eqref{inverse} by convex minimization schemes with some sparsity constraints.
Due to the energy-preserving property and computational efficiency, orthonormal bases and tight frames are often used. For example, the orthonormal basis in the discrete cosine transform (DCT) has many applications in image processing and has been known to be effective for sparsely approximating texture part of an image. Tight frames built from wavelets (called tight framelets in this paper) or from curvelets/shearlets are claimed to provide sparse approximations for natural images and they are 
particularly attractive for capturing the cartoon part of an image (\cite{cai2008framelet,candes2004new,HanACHA2012,HanZhaoSIIMS2014,KL,lim2013nonseparable}).

A tight frame generalizes an orthonormal basis by allowing redundancy.
Comparing with an orthonormal basis, a tight frame is an overcomplete system but keeps the energy-preserving property of an orthonormal basis.
Let $(\cH, \la \cdot, \cdot\ra)$ be a Hilbert space and $\mathcal{J}$ be an at most countable index set. A set
$\{\ws_j \in \cH \setsp j\in \mathcal{J} \}$ of elements in $\cH$ is called a tight frame for $\cH$ if the following energy-preserving identity holds:
\[
\la f, f\ra=\sum_{j\in \mathcal{J}} |\la f, \ws_j\ra|^2, \quad \mbox{for all} \quad f\in \cH.
\]
It follows directly from the above identity that for every $f\in \cH$, $f=\sum_{j\in \mathcal{J}} \la f, \ws_j\ra \ws_j$
with the series converging unconditionally in $\cH$.
In this paper we mainly deal with digital images which are finite dimensional.
Therefore, we are mainly interested in tight frames for the finite dimensional Euclidean space $\R^d$. Let $\ws_1, \ldots, \ws_n$ be column vectors in $\R^d$. Define a $d\times n$ synthesis/reconstruction matrix $\ws:=[\ws_1, \ldots, \ws_n]$, that is, the $j$th column of the matrix $\ws$ is the vector $\ws_j$.
Then $\{\ws_1, \ldots, \ws_n\}$ is a tight frame for $\R^d$ if and only if $\ws \ws^\tp=I$, where $\ws^\tp$ is the transpose of the matrix $\ws$. Hence, for simplicity, we say that a $d\times n$ real-valued matrix $\ws$ is a tight frame if $\ws \ws^\tp =I$.
The column vectors $\ws_j$ of $\ws$ are called frame elements.
If $n=d$, then $\ws$ is a square matrix and a tight frame $\{\ws_1, \ldots, \ws_n\}$ with $n=d$ for $\R^d$ is simply an orthonormal basis for $\R^d$. We define the redundancy rate of a tight frame $\ws\in \R^{d\times n}$ to be the ratio $n/d$, since it transforms a signal $\x\in \R^d$ of length $d$ into a coefficient vector $\ws^\tp \x\in \R^n$ of length $n$. Since $n=d$ for an orthonormal basis $\ws$, it has the redundancy rate $1$, that is, it has no redundancy at all.
Since we will deal with complex tight framelets in this paper, it is worth pointing out the following trivial fact: If $\{\ws_1, \ldots, \ws_n\}\subseteq \mathbb{C}^d$ is a tight frame for $\mathbb{C}^d$, then $\{\ws_1^{[r]}, \ldots, \ws_n^{[r]}, \ws_1^{[i]}, \ldots, \ws_n^{[i]}\}$ is a tight frame for $\R^d$, where $h^{[r]}$ and $h^{[i]}$ in $\R^d$ are the real and imaginary parts of a vector $h\in \mathbb{C}^d$ satisfying $h=h^{[r]}+ih^{[i]}$ with $i$ being the imaginary unit.

The inner product of two vectors $\x=(x_1, \ldots, x_d)^\tp, \y=(y_1, \ldots, y_d)^\tp \in \C^d$ is defined to be
$\la \x, \y  \ra = \sum_{j=1}^d x_j \overline{y_j}$, where $\overline{y_j}$ is the complex conjugate of the complex number $y_j\in \C$.
Since we often deal with $\x, \y\in \R^d$, $\la \x, \y  \ra$ is simply $\sum_{j=1}^d x_j y_j$.
For $0<p<\infty$, the $l_p$-norm of $\x=(x_1, \ldots, x_d)^\tp\in \R^d$ is defined by
$\|\x\|_p:=\left(\sum_{j=1}^d |x_j|^p\right)^{1/p}$.
Moreover, its $l_\infty$-norm is defined by $\|\x\|_\infty:=\max_{1\le j\le d} |x_j|$ and its $l_0$-norm
$\|\x\|_0$ is simply the number of nonzero numbers in $x_1, \ldots, x_d$. Note that $\|\cdot\|_p$ is only a quasi-norm for $0<p<1$.
If a $d\times n$ matrix $\ws$ is a tight frame, for any $\x=(x_1, \ldots, x_d)^\tp \in \R^d$, by $\ws \ws^\tp=I$, we have
\[
\x=\ws\ws^\tp \x=\ws \mc=\sum_{j=1}^n c_j \ws_j \quad \mbox{with}\quad
\mc=(c_1, \ldots, c_n)^\tp:=\ws^\tp \x.
\]
That is, $c_j=\la \x, \ws_j\ra, j=1, \ldots, n$ are the frame coefficients of $\x$ under a tight frame $\ws$. Thus, $\ws$ is the synthesis/reconstruction matrix/operator, while $\ws^\tp$ serves as the analysis/decomposition matrix/operator.
An image $\x$ often has sparse frame coefficients under some tight frames, i.e., $\|\mc\|_p$ is small for $0\le p\le 1$.

Let $\Omega\subseteq \{1, \ldots, d\}$ be a nonempty given observable region. Define a $d\times d$ diagonal matrix $\cP$ by
\be \label{proj}
[\cP]_{j,k}=
\begin{cases} 1, &\text{$j=k$ with $j\in \Omega$,}\\
0, &\text{$j=k$ with $j\not \in \Omega$ or $j\ne k$,}
\end{cases}\qquad j,k=1, \ldots, d.
\ee
Quite often, the observation $\y=(y_1, \ldots, y_d)^\tp$ in \eqref{inverse} is only available on a subset $\Omega$ of $\{1, \ldots, d\}$, while $y_j, j\not\in \Omega$ are missing and not available. The set $\{1, \ldots, d\}\backslash\Omega$ is often called an inpainting mask since data on this set is masked and not observable.
From now on, $\ws_j$ always denotes the $j$th column of a $d\times n$ matrix $\ws$ for $j=1, \ldots, n$. We always assume that $\Omega$ is a nonempty subset of $\{1, \ldots, d\}$ representing a given observable region of an image and $\cP$ is the $d\times d$ diagonal matrix defined in \eqref{proj}.
The underlying idea is that the unknown ``true'' clean image $\x$ in \eqref{inverse} has a sparse representation under a given tight frame $\ws$. More precisely, the frame coefficient vector $\mc=\ws^\tp \x$ of $\x$ is sparse, that is, $\|\mc\|_0$ is small. Since $l_0$-norm is not convex, one often uses the $l_1$-norm instead and expects that $\|\mc\|_1$ is small for the frame coefficient vector $\mc$ of a natural image $\x$. Now the solution $\x$ to \eqref{inverse} may be recovered through a convex minimization scheme with sparse constraints on its frame coefficients.
Explicitly,
given an observed image $\cP \y$ and an observation matrix $\cP \cA$, we define
\be \label{Bb}
\cB:=\cP \cA \quad \mbox{and}\quad \mb:=\cP \y.
\ee
Representing the unknown clean image $\x\in \R^d$ in the transform domain using a tight frame $\ws$ as $\x=\ws \mc$ with $\mc\in \R^n$, we see that the linear inverse/regression problem in \eqref{inverse} in the transform domain becomes
\be \label{inverse:2}
\mb=\cB \ws \mc +\n.
\ee
The following model has been proposed in the literature \cite{cai2008framelet} to solve the linear inverse/regression problem in \eqref{inverse} through \eqref{inverse:2} to restore the unknown clean image $\x$ as follows:
\be\label{balanced}
\min_{\mc\in\R^n} \frac{1}{2}\| \cB \ws \mc - \mb \|_2^2 + \lambda\| \mc \|_1
+ \kappa\|(I - \ws^\tp \ws)\mc\|_2^2 ,
\ee
where both $\lambda$ and $\kappa$ are nonnegative regularization parameters. Let $\hmc$ be a minimization solution to \eqref{balanced}, that is, $\hmc$ is a minimizer to \eqref{balanced}. Then a reconstructed image
$\hmx:=\ws \hmc$
is regraded as the recovered/restored image of the unknown clean image $\x$.
Model \eqref{balanced} is called the balanced approach in \cite{cai2008framelet}. A brief introduction on several frame-based iterative thresholding image inpainting algorithms using the balanced approach will be provided in Section~\ref{sec:related}. If $\kappa =0$, then the model \eqref{balanced} becomes the \emph{synthesis-based approach}:
\be\label{synthesis}
\min_{\mc\in\R^n} \frac{1}{2}\| \cB \ws \mc - \mb \|_2^2 + \lambda\| \mc \|_1.
\ee
If $\kappa = \infty$, then the last term in the model \eqref{balanced} must be equal to $0$, that is, $(I-\ws^\tp \ws)\mc=0$. Since $\ws \ws^\tp=I$, it is easy to see that $(I-\ws^\tp \ws)\mc=0$ if and only if $\mc \in \range(\ws^\tp):=\{\ws^\tp \x \setsp \x\in \R^d\}$. Moreover, the mapping $\R^d \rightarrow \range(\ws^\tp)$ with $\x\mapsto \ws^\tp \x$ is a bijection. Therefore, for $\kappa=\infty$, the balanced approach in \eqref{balanced} becomes the \emph{analysis-based approach}:
\be \label{analysis}
\min_{\mc \in \range(\ws^\tp)} \frac{1}{2} \| \cB \ws \mc - \mb \|_2^2 + \lambda \| \mc \|_1
=\min_{\x \in \R^d}
\frac{1}{2} \|  \cB \x - \mb \|_2^2 + \lambda \| \ws^\tp \x \|_1.
\ee
If $n=d$, then a tight frame $\ws\in \R^{d\times n}$ is a square matrix satisfying $\ws\ws^\tp=\ws^\tp \ws=I$ and hence the set of column vectors of $\ws$ forms an orthonormal basis for $\R^d$. For this particular case, the balanced approach,  the synthesis-based approach, and the analysis-based approach are the same.
However, these three approaches are indeed different when $n>d$ since $\ws$ is a redundant tight frame for $\R^d$. For more details on these three approaches, see \cite{cai2008framelet,cai2009split,elad2007analysis,shen2010wavelet} and many references therein.

Since a tight frame $\ws$ is a redundant system, frame elements in $\ws$ are often highly correlated and are close to each other. When two frame elements $\ws_j, \ws_k$ in a tight frame $\ws$ are close to each other, for $\x\in \R^d$, it is trivial to see that
$$
|\la \x, \ws_j\ra-\la \x, \ws_k\ra|\le \|\x\|_2 \|\ws_j-\ws_k\|_2.
$$
That is, their corresponding frame coefficients $\la \x, \ws_j\ra$ and $\la \x, \ws_k\ra$ are also close to each other. This is called the grouping effect in statistics. See Section~\ref{sec:ge} for more details. For effectiveness, it is important that the recovered coefficient vector $\hmc$ through the convex minimization scheme in \eqref{balanced} also enjoys this grouping effect property. We have the following result on grouping effect of the balanced approach in \eqref{balanced}.

\begin{theorem}\label{thm:ge}
Let $\ws$ be a $d\times n$ real-valued matrix (not necessarily satisfies the tight frame condition $\ws \ws^\tp =I$) and let
$\hmc = (\breve{c}_1,\ldots, \breve{c}_n)^\tp\in \R^n$ be a minimizer to the convex minimization problem in \eqref{balanced} with positive regularization parameters $\lambda$ and $\kappa$. Then
\be \label{ge:relation}
|\breve{c}_j - \breve{c}_k| \le  \Big\| (\ws_j - \ws_k)^\tp \Big(
(2I-\ws \ws^\tp)\ws \hmc-\tfrac{1}{2\kappa}\cB^\tp (\cB \ws \hmc-\mb)\Big)\Big\|_2
\ee
holds for all $j,k=1, \ldots, n.$
\end{theorem}

A proof of Theorem~\ref{thm:ge} will be given in
Section~\ref{sec:ge}.
The grouping effect property in Theorem~\ref{thm:ge} partially explains the effectiveness of the balanced approach \eqref{balanced}
for several image restoration problems in \cite{cai2008framelet,cai2010simultaneous,cai2009split}.
We shall further generalize Theorem~\ref{thm:ge} in Section~\ref{sec:ge} by replacing the term $\lambda \|\mc\|_1$ for sparsity in \eqref{balanced} with a more general term
$\|\diag(\lambda_1, \ldots, \lambda_n) \mc\|_1$.
Such a modified model for the balanced approach in \eqref{balanced} using a nonuniform weight vector $\diag(\lambda_1, \ldots, \lambda_n)$ instead of a uniform weight $\lambda$ has been widely used in image processing, see Section~\ref{sec:related} for more details.

For image inpainting, the matrix $\cA=I$ and the observed image $\y=(y_1, \ldots, y_d)^\tp\in \R^d$ is given by
\be \label{iip}
y_j=\begin{cases} x_j  + n_j, &\text{$j\in \Omega$,}\\
\mbox{arbitrary}, &  j\in \{ 1,\ldots,d\}\backslash \Omega,
\end{cases}
\ee
where $\Omega$ is a given observable region, $\x=(x_1, \ldots, x_d)^\tp \in \R^d$ is an unknown clean image to be restored, and $n_j$ is the noise term which is often the additive independent identically distributed (i.i.d.) zero-mean Gaussian noise.
The goal of image inpainting is to recover the missing pixels of $\x$ in the missing region $\{1, \ldots, d\}\backslash \Omega$ (i.e., the inpainting mask) while to recover the pixels of $\x$ on $\Omega$ by suppressing the noise of $\y$ in the observable region $\Omega$.
Many techniques have been proposed in the literature to address the image inpainting problem, for example, see \cite{Bertalmio2014,Bertalmio2003,cai2008framelet,cai2010simultaneous,cai2009split,Chan2005,elad:book,elad2005simultaneous,li2013adaptive,lim2013nonseparable} and many references therein.
According to the survey article \cite{Bertalmio2014}, methods for image inpainting can be classified into three groups: patch-based methods, PDEs/variational methods,
and sparse representation based methods. We shall briefly review several image inpainting methods in Section~\ref{sec:related}.
The readers can see  \cite{Bertalmio2014} and the books \cite{Chan2005,elad:book}
for a detailed list of references for many different methods to study
the image inpainting problem. In this paper, we concentrate on frame-based iterative image inpainting algorithms which belong to sparse representation based methods.

To solve the image inpainting problem in \eqref{iip} using convex minimization schemes as in \eqref{balanced}, many frame-based iterative thresholding algorithms have been proposed in the literature (\cite{cai2008framelet,cai2010simultaneous,cai2009split,elad2005simultaneous,li2013adaptive,lim2013nonseparable}).
Since real-world natural images are known to consist of two layers: cartoons and textures,
to improve the performance of an iterative thresholding algorithm for the image inpainting problem, quite often one uses two tight frames $\ws_c$ and $\ws_t$ instead of a single tight frame $\ws$ in \eqref{balanced}. More specifically, a tight frame $\ws_c$ is used so that the cartoon part of an image has a sparse representation under the tight frame $\ws_c$, while the texture part of an image has a good sparse approximation under the tight frame $\ws_t$. This is called simultaneous cartoon and texture inpainting in the literature. For more details on this and other related models, see Section~\ref{sec:related} and \cite{cai2008framelet,cai2010simultaneous,cai2009split,elad2005simultaneous,li2013adaptive,lim2013nonseparable}.
The commonly used tight frames $\ws_c$ for cartoons are tight framelets (that is, tight frames built from wavelet filters), curvelets, or shearlets etc. For textures, the tight frame derived from discrete cosine transform (DCT) is one of the most popular choices for $\ws_t$. Though the approach of simultaneous cartoon and texture inpainting is natural and interesting, as pointed out in \cite{li2013adaptive}, there are several unresolved issues about this approach. In particular, there is no precise mathematical definitions of cartoon and texture components of an image and therefore, it is difficult in terms of mathematics in separating an image exactly into its cartoon part and texture part. Moreover, there is currently no automatic way of determining the regularization parameters in such an approach and the performance of such image inpainting algorithms largely depends on the choices of the regularization parameters (see \cite{cai2008framelet,cai2010simultaneous,cai2009split,elad2005simultaneous}).

Inspired by recent development on directional tensor product complex tight framelets ($\tpctf$s) and their impressive performance for the image denoising problem in \cite{HanMMNP2013,HanZhaoSIIMS2014}, in this paper we propose an iterative thresholding algorithm using a single tight frame derived from $\tpctf$s for the image inpainting problem. Details on directional tensor product complex tight framelets $\tpctf$s will be provided in Section~\ref{sec:tpctf}.
Basically, our iterative thresholding algorithm is similar to many frame-based image inpainting algorithms and takes the following form: Let $\ws\in \R^{d\times n}$ be the tight frame built from the tensor product complex tight framelet $\tpctf_6$ and $\cP$ be the projection operator in \eqref{proj} with a given observable region $\Omega$. Given a corrupted image $\y=(y_1, \ldots, y_d)^\tp$ on the observable region $\Omega$ as in \eqref{iip}.
Assume that the noise standard deviation $\sigma$ of additive zero-mean i.i.d. Gaussian noise in \eqref{iip} is known. The following iterative algorithm (see Algorithm~\ref{mainalg} for full details) is proposed in this paper for image inpainting:

\begin{algorithm*}[htb]
\begin{algorithmic}[1]
\STATE Initialization: $\x_0 = 0$, $\lambda = \lambda_0$, $\ell=0$.
\WHILE{not convergent}
\STATE  $\mc_{\ell+1} = \mbox{Thresholding}_{\lambda} (\ws^\tp (\cP \y  + (I-\cP)\x_\ell))  $.  \label{step:thresholding}
\STATE  $\x_{\ell+1} = \ws\mc_{\ell+1}$.
\STATE  $\error = \| (I-\cP)(\x_{\ell+1} - \x_{\ell}) \|_2/\|  \cP \y\|_2$.
\IF    { $\error < \tolerence $}
\STATE  Update the thresholding value $\lambda$. \label{step:lambda}
\ENDIF
\STATE  $\ell=\ell+1$.
\ENDWHILE
\RETURN  $\x_{\ell+1}$ as the restored image.
\end{algorithmic}
\end{algorithm*}

\begin{figure}[th]
\centering
\subfigure[Barbara]{\includegraphics[width=2.5cm]{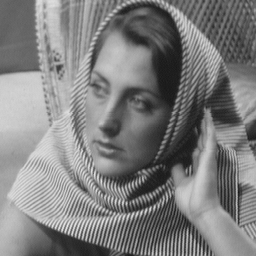}}
\subfigure[C.man]{\includegraphics[width=2.5cm]{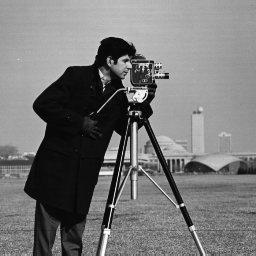}}
\subfigure[Lena]{\includegraphics[width=2.5cm]{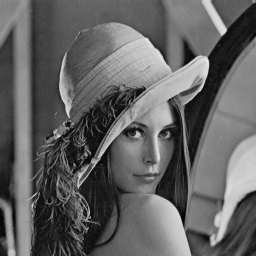}}
\subfigure[House]{\includegraphics[width=2.5cm]{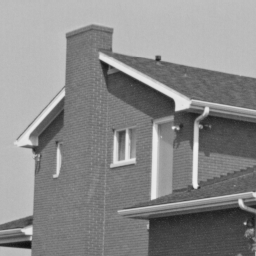}}
\subfigure[Peppers]{\includegraphics[width=2.5cm]{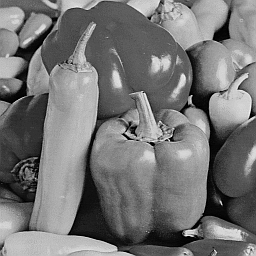}}\\
\subfigure[Hill]{\includegraphics[width=2.5cm]{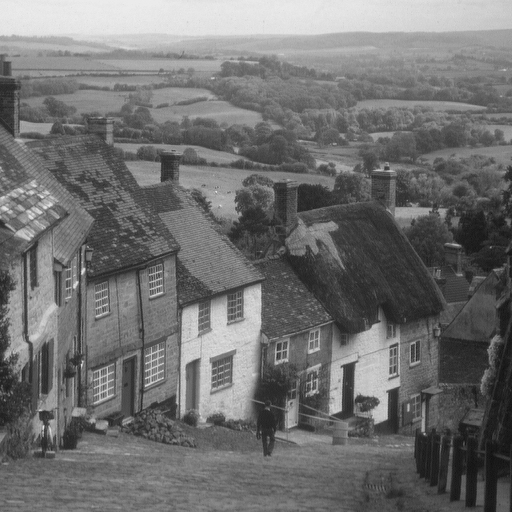}}
\subfigure[Man]{\includegraphics[width=2.5cm]{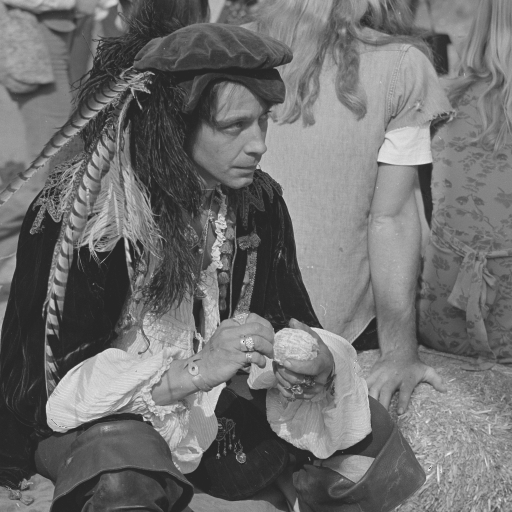}}
\subfigure[Boat]{\includegraphics[width=2.5cm]{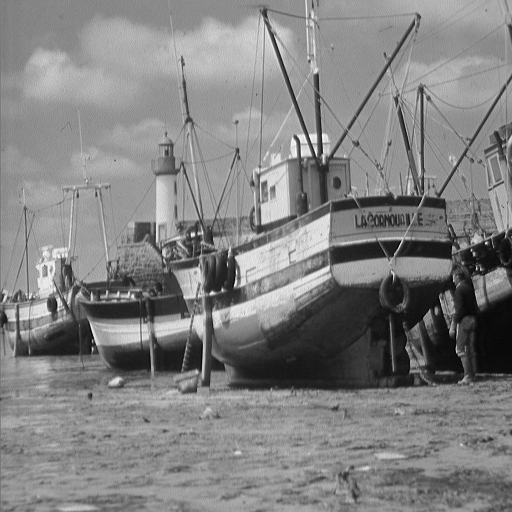}}
\subfigure[Barbara]{\includegraphics[width=2.5cm]{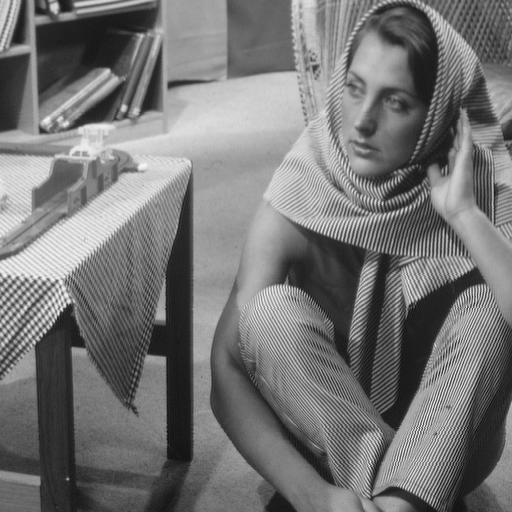}}
\subfigure[Mandrill]{\includegraphics[width=2.5cm]{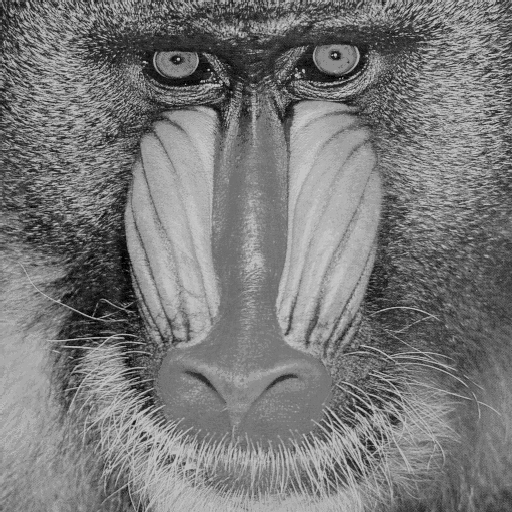}}\\
\subfigure[Text $1$]{\includegraphics[width=2.5cm]{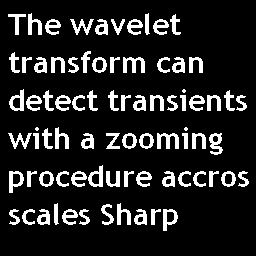}}
\subfigure[Text $2$]{\includegraphics[width=2.5cm]{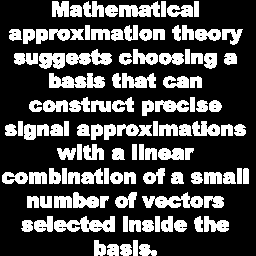}}
\subfigure[Text $3$]{\includegraphics[width=2.5cm]{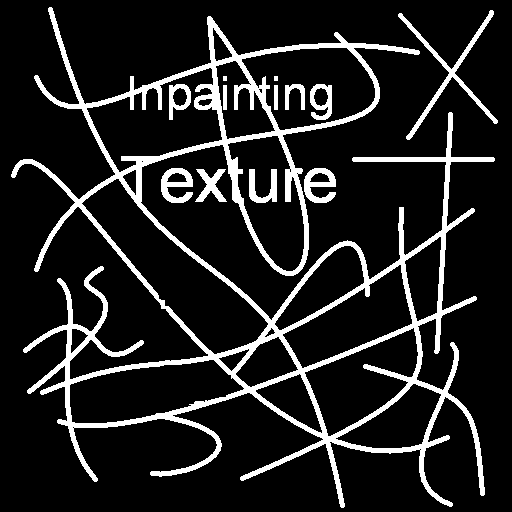}}
\subfigure[Text $4$]{\includegraphics[width=2.5cm]{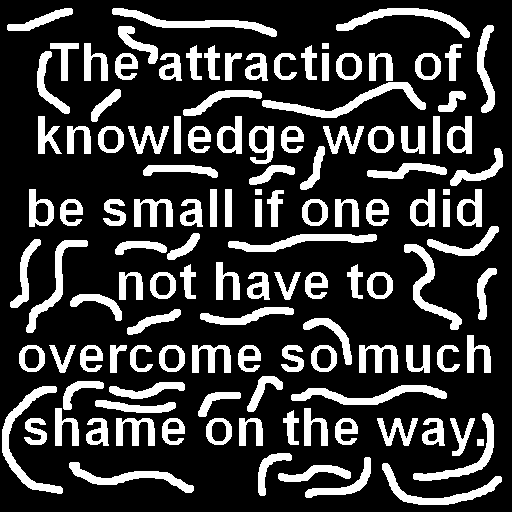}}
\caption{(a)-(e) are test images of size $256\times 256$. (f)-(j) are test images of size $512\times 512$. (k) and (l) are inpainting masks of size $256\times 256$.
(m) and (n) are inpainting masks of size $512\times 512$. The observable region $\Omega$ is just the complement of an inpainting mask.}\label{testimages}
\end{figure}

To illustrate the performance of our proposed image inpainting algorithm and to compare it with several state-of-the-art frame-based image inpainting algorithms, we provide here two tables of experimental results.
Table~\ref{tab:nonoise} presents the numerical results in terms of PSNR (Peak Signal-to-Noise Ratio) values for the image inpainting problem without noise (that is, $\sigma=0$). We compare the performance of our proposed algorithm with several frame-based iterative image inpainting algorithms including spline tight framelets
based image inpainting algorithm in \cite{cai2008framelet}, simultaneous cartoon and texture inpainting by combining spline tight framelets and local DCT in \cite{cai2010simultaneous,cai2009split},
morphological component analysis (MCA) based simultaneous cartoon and texture image inpainting in \cite{elad2005simultaneous},
adaptive inpainting algorithm based on undecimated transform using the DCT-Haar wavelet filters in \cite{li2013adaptive}, image inpainting algorithm based on undecimated transform using compactly supported nonseparable shearlets in \cite{lim2013nonseparable}.
The implementation of all those related frame-based image inpainting algorithms are kindly provided by their own authors or downloaded from their homepages. We run all the related image inpainting algorithms with their default parameter values which have been given in the source codes by their authors.
From Table~\ref{tab:nonoise} for the image inpainting problem without noise (that is, the noise standard deviation $\sigma=0$), we see that our proposed image inpainting algorithm often performs better than the well-known frame-based iterative image inpainting algorithms in \cite{cai2008framelet,cai2010simultaneous,cai2009split,elad2005simultaneous,li2013adaptive,lim2013nonseparable}.

\begin{table}[htbp]
\centering
\caption{
Performance in terms of PSNR values of several iterative image inpainting algorithms without noise for five $256\times 256$ test images in (a)-(e) of Figure~\ref{testimages}.
\cite{cai2008framelet} uses undecimated spline tight framelets, \cite{cai2010simultaneous,cai2009split} uses both undecimated spline tight framelets and local DCT for simultaneous cartoon and texture inpainting, \cite{elad2005simultaneous} uses MCA-based simultaneous cartoon and texture image inpainting, \cite{li2013adaptive} uses undecimated DCT-Haar wavelet filters, and \cite{lim2013nonseparable} uses undecimated compactly supported nonseparable shearlets.
Top left table for inpainting
mask Text~1 in (k) of Figure~\ref{testimages}.
Top right table for inpainting
mask Text~2 in (l) of Figure~\ref{testimages}.
Bottom left table for
$50\%$ randomly missing pixels.
Bottom right table for
$80\%$ randomly missing pixels.
Gain refers to the PSNR gain of our proposed Algorithm~\ref{mainalg} in row 7 over the highest PSNR values (underlined) among all the PSNR values in the first 6 rows obtained by the $6$ inpainting algorithms in \cite{cai2008framelet,cai2010simultaneous,cai2009split,elad2005simultaneous,li2013adaptive,lim2013nonseparable}.}

\begin{tabular}{|c|c|c|c|c|c||c|c|c|c|c|}
\hline
&Barbara &C.man &Lena &House &Peppers &Barbara &C.man &Lena &House &Peppers \\
\hline
\cite{cai2008framelet} & 32.82  & 30.28  & 32.11  & 37.08  & 31.41 & 27.77  & 26.75  & 28.32  & 31.01  & 27.10 \\
\cite{cai2010simultaneous} & 35.23  & 30.57  & 31.85  & 37.42  & 31.09 & 29.36  & 26.88  & 28.15  & 31.33  & 26.80 \\
\cite{cai2009split} & 34.98  & 30.48  & 31.91  & 37.54 & 31.10 & 29.01  & 26.70  & 28.05  & 31.18  & 26.76\\
\cite{elad2005simultaneous} & 35.47  & 31.12  & 33.81  & 37.66 & 31.44 & 29.55  & 27.26  & 29.21  & 33.26  & 27.56\\
\cite{li2013adaptive}  & \underline{36.60}  & 31.99  & 33.82  & 35.76  & 32.27 & 30.46  & \underline{28.58}  & 29.50  & 32.45  & 28.05\\
\cite{lim2013nonseparable} & 36.09  & \underline{32.30}  &\underline{34.89}  & \underline{39.68}
&\underline{32.49}
& \underline{30.56}  &28.52  &\underline{30.26}  &\underline{33.65}  &\underline{29.17}\\
Alg. \ref{mainalg} &
    \textbf{38.43} & \textbf{32.77} & \textbf{35.06} & \textbf{39.73} & \textbf{33.23} & \textbf{32.16} & \textbf{28.85} & \textbf{30.73} & \textbf{34.76} & \textbf{29.23} \\
\hline
Gain&1.83  & 0.47  & 0.17  & 0.05  & 0.74  & 1.60  & 0.27  & 0.47  & 1.12  & 0.06  \\
\hline \hline

\cite{cai2008framelet} & 29.58  & 28.65  & 31.39  & 36.57  & 29.18  & 24.34  & 23.94  & 26.27  & 29.80  & 24.67  \\
    \cite{cai2010simultaneous} & 32.59  & 28.32  & 30.85  & 36.90  & 28.74  & 26.27  & 23.75  & 25.88  & 29.73  & 24.35  \\
      \cite{cai2009split} & 32.68  & 28.39  & 30.82  & 36.93  & 28.94  & 25.90  & 23.69  & 25.85  & 29.54  & 24.36  \\
    \cite{elad2005simultaneous} & 31.78  & 27.63  & 30.57  & 34.90  & 27.83  & 25.01  & 23.31  & 26.33  & 29.17  & 23.73  \\
    \cite{li2013adaptive}  &\underline{34.97}  &\underline{30.03}  & \underline{32.63}  &\underline{38.51}  &\underline{29.42}  &\underline{26.85}  & 23.84  & 26.26  & 30.16  & 24.12  \\
     \cite{lim2013nonseparable} & 33.33  & 29.18  & 32.24  & 37.26  & 29.34  & 26.71  &\underline{24.72}  &\underline{27.77}  &\underline{31.41}  &\underline{25.37}  \\
    Alg. \ref{mainalg} &\textbf{36.25} & \textbf{30.31} & \textbf{33.60} & \textbf{39.24} & \textbf{30.31} & \textbf{28.22} & \textbf{25.09} & \textbf{28.15} & \textbf{32.31} & \textbf{25.66} \\
    \hline
  Gain &  1.28  & 0.28  & 0.97  & 0.73  & 0.89  & 1.37  & 0.37  & 0.38  & 0.90  & 0.29 \\
     \hline
\end{tabular}%
\label{tab:nonoise}
\end{table}%

In Section~\ref{sec:related}, we shall provide details about the inpainting algorithms and their underlying tight frames given in \cite{cai2008framelet,cai2010simultaneous,cai2009split,elad2005simultaneous,li2013adaptive,lim2013nonseparable}.
Here we make some remarks about the experimental results in Table~\ref{tab:nonoise}.
The inpainting algorithms in \cite{cai2008framelet,li2013adaptive,lim2013nonseparable} and our proposed Algorithm~\ref{mainalg} are very similar to each other in that all of them are iterative thresholding algorithms using one tight frame.
The inpainting algorithms in \cite{cai2010simultaneous,cai2009split,elad2005simultaneous} use an additional tight frame for simultaneous cartoon and texture inpainting. Some details are as follows:

\begin{enumerate}
\item For \cite{cai2008framelet}, the tight frame is built from the tensor product of an undecimated real-valued spline tight framelet filter bank $\{a; b_1, \ldots, b_4\}$ and has redundancy rate $25$ with one decomposition level.

\item For \cite{li2013adaptive}, the tight frame is built from the undecimated DCT-Haar wavelet filter derived from the discrete cosine transform (DCT) with a block size $7\times 7$ and has redundancy rate $49$.

\item For \cite{lim2013nonseparable}, the (non-tight) frame is built from undecimated compactly supported nonseparable shearlets and has redundancy rate $49$ with $16, 16, 8, 8$ high-pass filters and one low-pass filter.

\item For our proposed Algorithm~\ref{mainalg}, the tight frame is built from tensor product complex-valued tight framelet filter banks in \cite{HanMMNP2013,HanZhaoSIIMS2014} and has redundancy rate $\frac{32}{3}\approx 10.67$ with $4$ decomposition levels.

\item For \cite{cai2010simultaneous,cai2009split},
the tight frame $\ws_c$ for the cartoon part is built from the tensor product of an undecimated real-valued spline tight framelet filter bank $\{a; b_1, b_2\}$, and the tight frame $\ws_t$ is from the $50\%$ overlapping discrete cosine transform.
The redundancy rate of $\ws_c$ is $9$ with one decomposition level and the redundancy rate of $\ws_t$ is $2$. Therefore, the total redundancy rate is $11$.

\item For \cite{elad2005simultaneous}, the tight frame $\ws_c$ for the cartoon part is built from curvelets and the tight frame $\ws_t$ is from the $50\%$ overlapping discrete cosine transform with a block size $32\times 32$. The redundancy rate of $\ws_c$ is $2.89$ and the redundancy rate of $\ws_t$ is $2$. Therefore, the total redundancy rate is $4.89$.
\end{enumerate}

For test images mainly having edges and few textures, \cite{cai2008framelet} using spline tight framelets often tends to be slightly better than \cite{li2013adaptive}, while \cite{lim2013nonseparable} using compactly supported nonseparable shearlets is better than \cite{li2013adaptive}.
However, for texture-rich images and for inpainting masks with low percentages of randomly missing pixels, \cite{li2013adaptive} using the DCT-Haar filter is often better than both \cite{cai2008framelet} and \cite{li2013adaptive}. By using an additional tight frame for textures, \cite{cai2010simultaneous,cai2009split} often improve
\cite{cai2008framelet} for texture-rich images, but tends to lose performance for images having few textures.
Except using an extra total-variation (TV) regularization term,
the algorithm in \cite{elad2005simultaneous} is similar to those in \cite{cai2010simultaneous,cai2009split} and tends to perform better than \cite{cai2010simultaneous,cai2009split}.
Our Algorithm~\ref{mainalg} using complex-valued tight framelets in \cite{HanMMNP2013,HanZhaoSIIMS2014} often outperforms
\cite{cai2008framelet,cai2010simultaneous,cai2009split,elad2005simultaneous,li2013adaptive,lim2013nonseparable}.
The comparison of averaged performance in terms of averaged PSNR values is presented in Table~\ref{tab:average256}.

\begin{table}[htbp]
  \centering
  \caption{Each (averaged) PSNR value here is the average of the PSNR values in Table~\ref{tab:nonoise} for the five $256\times 256$ test images for each inpainting algorithm. (Averaged) Gain here refers to the gain of the averaged PSNR value of our Algorithm~\ref{mainalg} over the averaged PSNR value of the corresponding inpainting algorithm.}  \label{tab:average256}%
    \begin{tabular}{|c|cc|cc|cc|cc|}
    \hline
      & Text 1  & Gain & Text 2 & Gain & 50\% Missing & Gain & 80\% Missing & Gain \\
    \hline

     \cite{cai2008framelet} & 32.74  & 3.10  & 28.19  & 2.96  & 31.07  & 2.87  & 25.81  & 2.07  \\
    \cite{cai2010simultaneous} & 33.23  & 2.61  & 28.50  & 2.65  & 31.48  & 2.46  & 26.00  & 1.88 \\
      \cite{cai2009split} & 33.20  & 2.64  & 28.34  & 2.81  & 31.55  & 2.39  & 25.87  & 2.01  \\
    \cite{elad2005simultaneous} & 33.90  & 1.94  & 29.37  & 1.78  & 30.54  & 3.40  & 25.51  & 2.37  \\
    \cite{li2013adaptive}  & 34.09  & 1.75  & 29.81  & 1.34  & 33.11  & 0.83  & 26.24  & 1.64  \\
     \cite{lim2013nonseparable} & 35.09  & 0.75  & 30.43  & 0.72  & 32.27  & 1.67  & 27.19  & 0.69  \\
    Alg. \ref{mainalg} & 35.84  &   & 31.15  &   & 33.94  &   & 27.88  &  \\
    \hline
    \end{tabular}%
\end{table}%

Table~\ref{tab:noise} provides the numerical results in terms of PSNR values for the image inpainting problem with additive i.i.d. zero-mean Gaussian noise having noise standard deviation $\sigma$ ranging from $5$ to $50$.
For the image inpainting problem with noise (in particular, when the noise deviation $\sigma$ is relatively large), all other available frame-based image inpainting algorithms
in \cite{cai2008framelet,cai2010simultaneous,cai2009split,elad2005simultaneous,li2013adaptive,lim2013nonseparable}
with their default choices of parameter values often produce much inferior PSNR values than ours. This is probably largely due to the facts: (1) The image inpainting algorithms have been developed mainly for the noiseless case without considering noise (\cite{li2013adaptive,lim2013nonseparable}). (2) For image inpainting algorithms which take into account of noise or which are capable of handling noise,
for different levels of noise and different images, one has to manually tune several regularization parameters to achieve reasonable performance for the image inpainting problem with noise (\cite{cai2008framelet,cai2010simultaneous,cai2009split,elad2005simultaneous}). (3) Most known frame-based image inpainting algorithms may work well when the standard deviation $\sigma$ is very small but often fail when $\sigma$ is large, say $\sigma=30$, despite fine tuning of regularization parameters.
Except that the noise standard deviation $\sigma$ and the inpainting mask are assumed to be known (as assumed by all other frame-based image inpainting algorithms),
our proposed image inpainting algorithm does not require to manually tune parameters for reasonably good performance. To avoid potential unfair or misleading comparison with other available frame-based image inpainting algorithms for the image inpainting problem with noise, we only provide numerical results for our algorithm in Table~\ref{tab:noise}, which clearly shows that our image inpainting algorithm using $\tpctf$s performs superior for image inpainting with noise.
In Section~\ref{sec:experiment}, we shall provide all details on our proposed
Algorithm~\ref{mainalg} for image inpainting.
More experimental results on image inpainting with or without noise will be reported in Section~\ref{sec:experiment}. Comparison of our proposed image inpainting algorithm with the frame-based image inpainting algorithms in \cite{cai2008framelet,cai2010simultaneous,cai2009split,elad2005simultaneous,li2013adaptive,lim2013nonseparable}, in terms of both PSNR values and visual effects, will be given in Section~\ref{sec:experiment}.

\begin{table}[hbtp]
\centering
\caption{Performance in terms of PSNR values of our proposed image inpainting algorithm
under i.i.d. zero-mean Gaussian noise with noise standard deviation $\sigma=5, 10, 20, 30, 50$
for five $256\times 256$ test images in (a)-(e) of Figure~\ref{testimages}.
Top left table for inpainting mask Text~1 in (k) of Figure~\ref{testimages}.
Top right table for inpainting mask Text~2 in (l) of Figure~\ref{testimages}.
Bottom left table for $50\%$ randomly missing pixels.
Bottom right table for $80\%$ randomly missing pixels.
}
\begin{tabular}{|c|c|c|c|c|c||c|c|c|c|c|}
    \hline
    $\sigma$  & Barbara & C.man & Lena & House & Peppers & Barbara & C.man & Lena & House & Peppers \\
       \hline
    $5$  & 34.95  & 31.52  & 33.12  & 36.18  & 31.44  & 30.78  & 28.31  & 29.88  & 33.25  & 28.45  \\
    $10$ & 32.32  & 29.97  & 31.32  & 33.85  & 29.73  & 29.20  & 27.39  & 28.77  & 31.74  & 27.52  \\
    $20$ & 29.16  & 27.73  & 28.85  & 31.33  & 27.57  & 27.10  & 26.01  & 27.08  & 29.57  & 25.96  \\
    $30$ & 27.23  & 26.25  & 27.14  & 29.64  & 26.10  & 25.68  & 24.89  & 25.87  & 28.11  & 24.76  \\
    $50$ & 24.92  & 24.31  & 25.03  & 27.31  & 24.13  & 23.89  & 23.16  & 24.21  & 26.18  & 23.04  \\
    \hline\hline
    $5$  & 33.60  & 29.52  & 32.06  & 35.73  & 29.34  & 27.68  & 24.82  & 27.62  & 31.26  & 25.18  \\
    $10$ & 31.03  & 28.41  & 30.40  & 33.16  & 28.27  & 26.54  & 24.35  & 26.75  & 29.91  & 24.61  \\
    $20$ & 27.76  & 26.58  & 28.01  & 30.43  & 26.44  & 24.50  & 23.57  & 25.29  & 27.74  & 23.56  \\
    $30$ & 25.74  & 25.33  & 26.42  & 28.59  & 25.08  & 23.22  & 22.80  & 24.07  & 26.08  & 22.65  \\
    $50$ & 23.63  & 23.55  & 24.33  & 26.28  & 23.12  & 21.51  & 21.39  & 22.29  & 23.89  & 21.06  \\
     \hline
\end{tabular}%
\label{tab:noise}
\end{table}%

The structure of the paper is as follows. In Section~\ref{sec:ge} we shall prove and generalize Theorem~\ref{thm:ge} for the grouping effect property of the balanced approach in \eqref{balanced}.
We shall also briefly review the related elastic net in statistics and its grouping effect which inspires us to investigate the grouping effect of the balanced approach in this paper.
In Section~\ref{sec:related}, we shall review several related frame-based image inpainting algorithms including
spline tight framelet image inpainting in \cite{cai2008framelet}, simultaneous cartoon and texture image inpainting in \cite{cai2010simultaneous,cai2009split,elad2005simultaneous}, adaptive image inpainting using the DCT-Haar wavelet filters in \cite{li2013adaptive}, and image inpainting using compactly supported nonseparable shearlets in \cite{lim2013nonseparable}.
In Section~\ref{sec:tpctf} we shall outline the construction of directional tensor product complex tight framelets which have been initially introduced in \cite{HanMMNP2013} and fully developed in \cite{HanZhaoSIIMS2014}.
In Section~\ref{sec:experiment} we shall provide full details on our proposed iterative thresholding image inpainting algorithm using directional complex tight framelets. Detailed experimental results on the image inpainting problem with or without noise will be reported in Section~\ref{sec:experiment}.

\section{Grouping Effect of the Balanced Approach}
\label{sec:ge}

In this section we shall prove and generalize Theorem~\ref{thm:ge} for the grouping effect property of the balanced approach in \eqref{balanced}. Before presenting the proof and generalization to Theorem~\ref{thm:ge}, we briefly recall the elastic net in statistics and its grouping effect established in \cite{zhou2013grouping,zou2005regularization}.

Consider the linear inverse/regression model $\mb=E \mc+\n$, which is essentially the model \eqref{inverse:2} with $E=\cB \ws$.
The main task is to recover the unknown vector $\mc\in \R^n$ from the given observation $\mb \in \R^d$ and the given observation matrix $E \in\R^{d\times n}$.
Introduced in \cite{zou2005regularization},
the elastic net is a convex minimization scheme to achieve this goal and is given by
\be\label{elastic}
\min_{\mc\in\R^n} \frac{1}{2}\| E \mc - \mb \|_2^2
+ \lambda \|\mc \|_1 +\kappa \|\mc\|_2^2,
\ee
where $\lambda$ and $\kappa$ are positive regularization parameters.
Let $\hmc=(\breve{c}_1, \ldots, \breve{c}_n)^\tp\in \R^n$ be a minimizer to the convex minimization scheme in \eqref{elastic}. Then
the grouping effect property for the elastic net in \eqref{elastic} has been first proved in \cite{zou2005regularization} and then further improved in \cite[Theorem~3]{zhou2013grouping} as follows:
\be\label{groupe}
|\breve{c}_j- \breve{c}_k| \le \frac{1}{2\kappa}\|(E_j-E_k)^\tp (E\hmc-\mb)\|_2, \qquad \forall\; j,k=1, \ldots, n.
\ee
When the columns $E_j$ and $E_k$ of the matrix $E$ are close to each other,
the above inequality in \eqref{groupe} shows that the coefficients  $\breve{c}_j$ and $\breve{c}_k$ are also close to each other. This is called the grouping effect property of the elastic net in \cite{zhou2013grouping,zou2005regularization}. Since $\hmc$ is a minimizer to \eqref{elastic}, it is trivial to see that $\| E\hmc-\mb\|_2\le \|\mb\|_2$. Therefore, it follows directly from \eqref{groupe} that $|\breve{c}_j- \breve{c}_k|\le \frac{\|\mb\|_2}{2\kappa} \|E_j-E_k\|_2$ for all $j,k=1, \ldots, n$.

We notice that the balanced approach \eqref{balanced} differs to the elastic net \eqref{elastic} in that $E=\cB\ws$ and the last term $\kappa\|\mc\|_2^2$ in the elastic net \eqref{elastic} is replaced by $\kappa \|(I-\ws^\tp \ws)\mc\|_2^2$ in the balanced approach \eqref{balanced}. This motivates us to study whether a similar grouping effect property can be established for the balanced approach \eqref{balanced}.
Inspired by the work on the grouping effect property of elastic net \eqref{elastic} in \cite{zhou2013grouping,zou2005regularization}, we obtain in Theorem~\ref{thm:ge} the grouping effect property of the balanced approach \eqref{balanced}. We now prove Theorem~\ref{thm:ge}.

\begin{proof}[Proof of Theorem~\ref{thm:ge}]
As in \cite{zhou2013grouping,zou2005regularization}, for $\mc\in \R^n$, we define a function $F$ by
\[
F(\mc):=\frac{1}{2}\| \cB \ws \mc - \mb \|_2^2 + \lambda\|\mc \|_1+\kappa \|(I - \ws^\tp \ws)\mc\|_2^2.
\]
Let $e_j$ be the $j$th unit coordinate column vector of $\R^n$.
Let $\hmc=(\breve{c}_1,\ldots,\breve{c}_n)^\tp\in \R^n$ be a minimizer to the convex minimization problem in \eqref{balanced}.
For any $\gep\in \R$, by the definition of the function $F$, we have
\begin{align*}
F(\hmc + \gep e_j) - F(\hmc)
=&\la \cB \ws \hmc-\mb, \cB \ws e_j \ra\gep+
\frac{1}{2}\la \cB \ws e_j, \cB \ws e_j\ra \gep^2
+\lambda( |\breve{c}_j + \gep| -  |\breve{c}_j|)\\
&+2 \kappa \la (I-\ws^\tp \ws)\hmc, (I-\ws^\tp \ws) e_j\ra\gep
+\kappa \la (I-\ws^\tp \ws)e_j, (I-\ws^\tp \ws)e_j\ra \gep^2.
\end{align*}
%
Therefore, we have
\be \label{F:1}
F_\gep+
\Big(\frac{1}{2} \|\cB \ws e_j\|_2^2+\kappa \|(I-\ws^\tp \ws)e_j\|_2^2\Big)\gep^2=F(\hmc + \gep e_j) - F(\hmc)\ge 0,
\ee
where in the last inequality we used the fact that $\hmc$ is a minimizer to \eqref{balanced} and $F_\gep$ is defined to be
\[
F_\gep:=\la \cB \ws \hmc-\mb, \cB \ws e_j \ra\gep
+\lambda( |\breve{c}_j + \gep| -  |\breve{c}_j|)
+2 \kappa \la (I-\ws^\tp \ws)\hmc, (I-\ws^\tp \ws) e_j\ra\gep.
\]

We first consider the case $\breve{c}_j\ne 0$. Then we can take $\gep$ small enough so that $|\gep| < |\breve{c}_j|$. Consequently,
\[
|\breve{c}_j + \gep| -  |\breve{c}_j |  = \sgn(\breve{c}_j) \gep,
\]
where
the signum function of a real number $c\in \R$ is defined to be
$\sgn(c):=-1$ if $c<0$, $\sgn(c):=1$ if $c>0$, and $\sgn(0):=0$.
Hence, we can rewrite $F_\gep$ as follows:
\[
F_\gep=\Big( \la \cB \ws \hmc-\mb, \cB \ws e_j \ra
+\lambda \sgn(\breve{c}_j) +2 \kappa \la (I-\ws^\tp \ws)\hmc, (I-\ws^\tp \ws) e_j\ra\Big) \gep.
\]
Taking $\gep \to 0$, we conclude from \eqref{F:1} that $F_\gep=0$. Hence, under the assumption $\breve{c}_j\ne 0$, we must have
\be \label{relation}
\la \cB \ws \hmc-\mb, \cB \ws e_j \ra
+\lambda \sgn(\breve{c}_j) +2 \kappa \la (I-\ws^\tp \ws)\hmc, (I-\ws^\tp \ws)e_j\ra=0.
\ee
Since $\ws e_j=\ws_j$ is the $j$th column of $\ws$, we have
\be \label{eq1}
\la \cB \ws \hmc-\mb, \cB \ws e_j \ra=\la \cB^\tp (\cB \ws\hmc-\mb), \ws_j\ra=\ws_j^\tp \cB^\tp (\cB \ws\hmc-\mb)
\ee
and
\be \label{eq2}
\la (I-\ws^\tp \ws)\hmc, (I-\ws^\tp \ws) e_j\ra=
\la (I-\ws^\tp \ws)^2 \hmc, e_j\ra
=\breve{c}_j-\ws_j^\tp (2I-\ws \ws^\tp)\ws\hmc.
\ee
Rearranging the terms in the identity \eqref{relation}, by \eqref{eq1} and \eqref{eq2}, we conclude that
\[
\lambda \sgn(\breve{c}_j) + 2\kappa \breve{c}_j
=\ws_j^\tp \Big( 2 \kappa (2I-\ws \ws^\tp) \ws\hmc-\cB^\tp(\cB \ws \hmc-\mb)\Big).
\]
Since
$\breve{c}_j = \sgn(\breve{c}_j)|\breve{c}_j|$, the above identity becomes
\be\label{sgn1}
\ws_j^\tp \Big( 2\kappa (2I-\ws \ws^\tp) \ws\hmc-\cB^\tp(\cB \ws \hmc-\mb)\Big)
=\sgn(\breve{c}_j)(\lambda + 2\kappa |\breve{c}_j|).
\ee
%
%
%
We now consider the case $\breve{c}_j=0$. Then $|\breve{c}_j+\gep|-|\breve{c}_j|=\sgn(\gep)\gep$ and we have
\[
F_\gep=\Big( \la \cB \ws \hmc-\mb, \cB \ws e_j \ra
+\lambda \sgn(\gep) +2 \kappa \la (I-\ws^\tp \ws)\hmc, (I-\ws^\tp \ws) e_j\ra\Big) \gep.
\]
Taking $\gep\to 0$, we conclude from \eqref{F:1} that $F_\gep\ge 0$. In other words, we must have
\[
|\la \cB \ws \hmc-\mb, \cB \ws e_j \ra
+2 \kappa \la (I-\ws^\tp \ws)\hmc, (I-\ws^\tp \ws) e_j\ra|\le \lambda.
\]
By \eqref{eq1} and \eqref{eq2}, since $\breve{c}_j=0$, the above inequality can be rewritten as
\be \label{eq3}
\Big|\ws_j^\tp\Big(2\kappa (2I-\ws \ws^\tp) \ws\hmc-\cB^\tp(\cB\ws\hmc-\mb)
\Big)\Big|\le \lambda.
\ee

We complete the proof by considering four cases.

Case 1:  $\sgn(\breve{c}_j) \sgn(\breve{c}_k)>0$. Then \eqref{sgn1} leads to
\[
2\kappa(\breve{c}_j - \breve{c}_k) = (\ws^\tp_j - \ws^\tp_k)
\Big( 2\kappa (2I-\ws\ws^\tp) \ws\hmc-\cB^\tp(\cB \ws \hmc-\mb)\Big).
\]
Now it is trivial to see that \eqref{ge:relation} holds.

Case 2: $\sgn(\breve{c}_j)\sgn(\breve{c}_k)<0$. Then we have
\be\label{sgn2}
\sgn(\breve{c}_j - \breve{c}_k) = \sgn(\breve{c}_j).
\ee
It follows from \eqref{sgn1} and $\sgn(\breve{c}_k)=-\sgn(\breve{c}_j)$ that
\begin{align*}
(\ws^\tp_j - \ws^\tp_k) \Big(2\kappa (2I-\ws\ws^\tp) \ws \hmc-\cB^\tp(\cB \ws \hmc-\mb)\Big)
& = 2\kappa (\breve{c}_j - \breve{c}_k ) + \lambda{\sgn(\breve{c}_j)}
-\lambda {\sgn (\breve{c}_k)}\\
&=2\kappa (\breve{c}_j - \breve{c}_k) + 2\lambda \sgn(\breve{c}_j).
\end{align*}
We deduce from  \eqref{sgn2} and the above identity that
\[
|\breve{c}_j -\breve{c}_k| \le  \frac{2\kappa|\breve{c}_j - \breve{c}_k| + 2\lambda}{2\kappa} = \frac{1}{2\kappa}
\left|(\ws^\tp_j - \ws^\tp_k) \Big(2\kappa (2I-\ws \ws^\tp) \ws \hmc-\cB^\tp( \cB\ws \hmc-\mb)\Big)\right|.
\]
This proves \eqref{ge:relation}.

Case 3: $\sgn(\breve{c}_j)\ne 0$ and $\breve{c}_k=0$. Then $\breve{c}_k=0$ and \eqref{eq3} together imply
\[
\Big|\ws_k^\tp\Big(2\kappa (2I-\ws \ws^\tp) \ws\hmc
-\cB^\tp(\cB\ws\hmc-\mb)\Big)\Big|\le \lambda.
\]
Since $\sgn(\breve{c}_j)\ne 0$, \eqref{sgn1} holds and consequently,
\[
\left|
\ws_j^\tp \Big( 2\kappa (2I-\ws \ws^\tp) \ws\hmc-\cB^\tp(\cB \ws \hmc-\mb)\Big)
\right| = \lambda + 2\kappa |\breve{c}_j|.
\]
Therefore, we deduce from the above inequality and the above identity that
\begin{align*}
2\kappa |\breve{c}_j-\breve{c}_k|
&=2\kappa|\breve{c}_j|=\left|
\ws_j^\tp \Big( 2\kappa (2I-\ws \ws^\tp) \ws\hmc-\cB^\tp(\cB \ws \hmc-\mb)\Big)
\right|-\lambda\\
&\le \left|
\ws_j^\tp \Big( 2\kappa (2I-\ws \ws^\tp) \ws\hmc-\cB^\tp(\cB \ws \hmc-\mb)\Big)
\right|-\Big|\ws_k^\tp\Big(2\kappa (2I-\ws \ws^\tp) \ws\hmc-\cB^\tp(\cB\ws\hmc-\mb)\Big)\Big|\\
&\le \Big|(\ws_j^\tp-\ws_k^\tp)\Big(2\kappa (2I-\ws \ws^\tp) \ws\hmc-\cB^\tp(\cB\ws\hmc-\mb)\Big)\Big|.
\end{align*}
This proves \eqref{ge:relation}. The case $\sgn(\breve{c}_k)\ne 0$ and $\breve{c}_j=0$ can be proved exactly the same way by switching the role of $j$ and $k$.

Case 4: $\breve{c}_j=\breve{c}_k=0$. For this case, \eqref{ge:relation} trivially holds.

This completes the proof of Theorem~\ref{thm:ge}.
\end{proof}

We now explain that Theorem~\ref{thm:ge} recovers \cite[Theorem~3]{zhou2013grouping} for the grouping effect property of the elastic net in \eqref{elastic}.
Take $\ws=\sqrt{2}I$ and $ \cB=\frac{\sqrt{2}}{2}E$ in Theorem~\ref{thm:ge}. Then \eqref{balanced} becomes \eqref{elastic} and \eqref{ge:relation} of Theorem~\ref{thm:ge} becomes
\[
|\breve{c}_j-\breve{c}_k|
\le \tfrac{\sqrt{2}}{2\kappa} \| (\cB_j^\tp-\cB_k^\tp)(\sqrt{2}\cB \hmc-\mb)\|_2
=\tfrac{1}{2 \kappa} \| (E_j-E_k)^\tp (E\hmc-\mb)\|_2
\]
which is exactly \eqref{groupe} for the grouping effect property of the elastic net \eqref{elastic}.

The right-hand side of \eqref{ge:relation} in Theorem~\ref{thm:ge} can be further estimated without involving the minimizer $\hmc$.
Since $\hmc$ is a minimizer to \eqref{balanced}, we must have $\|\cB \ws\hmc-\mb\|_2\le \|\mb\|_2$ and $\|\hmc\|_2\le \|\hmc\|_1\le \frac{1}{2\lambda} \|\mb\|_2^2$. Consequently, for the expressions on the right-hand side of \eqref{ge:relation} in Theorem~\ref{thm:ge}, we have
\[
\|\tfrac{1}{2\kappa}(\ws_j-\ws_k)^\tp \cB^\tp(\cB\ws\hmc-\mb)\|_2
\le \tfrac{1}{2\kappa}
\|(\ws_j-\ws_k)^\tp \cB^\tp\|_2 \|\cB \ws\hmc-\mb\|_2\\
\le \tfrac{1}{2\kappa} \|\mb\|_2 \|(\ws_j-\ws_k)^\tp \cB^\tp\|_2
\]
and
\[
\| (\ws_j-\ws_k)^\tp (2I-\ws\ws^\tp)\ws \hmc\|_2
\le \|(\ws_j-\ws_k)^\tp (2I-\ws\ws^\tp)\ws\|_{2}\|\hmc\|_2
\le \tfrac{1}{2\lambda}\|\mb\|_2^2 \|(\ws_j-\ws_k)^\tp (2I-\ws\ws^\tp)\ws\|_{2}.
\]

As we shall see in Section~\ref{sec:related},
many image inpainting algorithms in the literature employ a more general model than \eqref{balanced} by replacing the term $\lambda \|\mc\|_1$ in \eqref{balanced} with
$\|\diag(\lambda_1, \ldots, \lambda_n) \mc\|_1$.
The following result generalizes Theorem~\ref{thm:ge} for such a generalized model.

\begin{theorem}\label{thm:ge:2}
Let $\ws$ be a $d\times n$ real-valued matrix (not necessarily satisfies the tight frame condition $\ws \ws^\tp =I$) and let
$\hmc = (\breve{c}_1,\ldots, \breve{c}_n)^\tp\in \R^n$ be a minimizer to the following convex minimization problem:
\be \label{balanced:2}
\min_{\mc\in \R^n} \frac{1}{2}\|\cB\ws\mc-\mb\|_2^2+\|\diag(\lambda_1, \ldots, \lambda_n)\mc\|_1+\kappa\|(I-\ws^\tp \ws)\mc\|_2^2
\ee
with positive real numbers $\lambda_1, \ldots, \lambda_n$ and $\kappa$. Then
\be \label{ge:relation:2}
|\breve{c}_j - \breve{c}_k| \le \tfrac{1}{2\kappa} |\lambda_j-\lambda_k|+ \Big\| (\ws_j - \ws_k)^\tp \Big(
(2I-\ws \ws^\tp)\ws \hmc-\tfrac{1}{2\kappa}\cB^\tp (\cB \ws \hmc-\mb)\Big)\Big\|_2
\ee
holds for all $j,k=1, \ldots, n$.
\end{theorem}

\begin{proof} As in the proof of Theorem~\ref{thm:ge}, we define
\[
F(\mc):=\frac{1}{2}\| \cB \ws \mc - \mb \|_2^2 + \|\diag(\lambda_1, \ldots, \lambda_n)\mc \|_1+\kappa \|(I - \ws^\tp \ws)\mc\|_2^2.
\]
By the same argument as in the proof of Theorem~\ref{thm:ge}, we see that \eqref{sgn1} and \eqref{eq3} hold with $\lambda$ being replaced by $\lambda_j$ for both cases of $\breve{c}_j\ne 0$ and $\breve{c}_j=0$, respectively.

We complete the proof by considering four cases as in the proof of Theorem~\ref{thm:ge}.

Case 1:  $\sgn(\breve{c}_j) \sgn(\breve{c}_k)>0$. Then \eqref{sgn1} with $\lambda=\lambda_j$ leads to
\[
2\kappa(\breve{c}_j - \breve{c}_k) = -\sgn(\breve{c}_j)(\lambda_j-\lambda_k)+(\ws^\tp_j - \ws^\tp_k)
\Big( 2\kappa (2I-\ws\ws^\tp) \ws\hmc-\cB^\tp(\cB \ws \hmc-\mb)\Big).
\]
Now it is trivial to see that \eqref{ge:relation:2} holds.

Case 2: $\sgn(\breve{c}_j)\sgn(\breve{c}_k)<0$. Then $\sgn(\breve{c}_k)=-\sgn(\breve{c}_j)$ and
it follows from \eqref{sgn1} with $\lambda=\lambda_j$ that
\[
(\ws^\tp_j - \ws^\tp_k) \Big(2\kappa (2I-\ws\ws^\tp) \ws \hmc-\cB^\tp(\cB \ws \hmc-\mb)\Big)
= 2\kappa (\breve{c}_j - \breve{c}_k ) +
\sgn(\breve{c}_j)(\lambda_j+\lambda_k).
\]
We deduce from  \eqref{sgn2} and the above identity that
\[
|\breve{c}_j -\breve{c}_k| \le  \frac{2\kappa|\breve{c}_j - \breve{c}_k| + \lambda_j+\lambda_k}{2\kappa} = \frac{1}{2\kappa}
\left|(\ws^\tp_j - \ws^\tp_k) \Big(2\kappa (2I-\ws \ws^\tp) \ws \hmc-\cB^\tp( \cB\ws \hmc-\mb)\Big)\right|.
\]
This proves \eqref{ge:relation:2}.

Case 3: $\sgn(\breve{c}_j)\ne 0$ and $\breve{c}_k=0$. Then $\breve{c}_k=0$ and \eqref{eq3} together imply
\[
\Big|\ws_k^\tp\Big(2\kappa (2I-\ws \ws^\tp) \ws\hmc
-\cB^\tp(\cB\ws\hmc-\mb)\Big)\Big|\le \lambda_k.
\]
Since $\sgn(\breve{c}_j)\ne 0$, \eqref{sgn1} with $\lambda=\lambda_j$ holds and consequently,
\[
\left|
\ws_j^\tp \Big( 2\kappa (2I-\ws \ws^\tp) \ws\hmc-\cB^\tp(\cB \ws \hmc-\mb)\Big)
\right| = \lambda_j + 2\kappa |\breve{c}_j|.
\]
Therefore, we deduce from the above inequality and the above identity that
\begin{align*}
&2\kappa |\breve{c}_j-\breve{c}_k|
=2\kappa|\breve{c}_j|=\left|
\ws_j^\tp \Big( 2\kappa (2I-\ws \ws^\tp) \ws\hmc-\cB^\tp(\cB \ws \hmc-\mb)\Big)
\right|-\lambda_k+\lambda_k-\lambda_j\\
&\le (\lambda_k-\lambda_j)+\left|
\ws_j^\tp \Big( 2\kappa (2I-\ws \ws^\tp) \ws\hmc-\cB^\tp(\cB \ws \hmc-\mb)\Big)
\right|-\Big|\ws_k^\tp\Big(2\kappa (2I-\ws \ws^\tp) \ws\hmc-\cB^\tp(\cB\ws\hmc-\mb)\Big)\Big|\\
&\le |\lambda_j-\lambda_k|+\Big|(\ws_j^\tp-\ws_k^\tp)\Big(2\kappa (2I-\ws \ws^\tp) \ws\hmc-\cB^\tp(\cB\ws\hmc-\mb)\Big)\Big|.
\end{align*}
This proves \eqref{ge:relation:2}. The case $\sgn(\breve{c}_k)\ne 0$ and $\breve{c}_j=0$ can be proved
by switching the role of $j$ and $k$.

Case 4: $\breve{c}_j=\breve{c}_k=0$. For this case, \eqref{ge:relation:2} trivially holds.
\end{proof}

When two frame elements $\ws_j$ and $\ws_k$ are at the same resolution/scale level, one often uses the same weight as $\lambda_j=\lambda_k$. Consequently, the presence of the term $|\lambda_j-\lambda_k|$ in \eqref{ge:relation:2} is often not an issue.

\section{Brief Review of Several Related Frame-based Image Inpainting Algorithms}
\label{sec:related}

In this section, we review some related image inpainting algorithms in the literature, in particular, the frame-based iterative image inpainting algorithms proposed in \cite{cai2008framelet,cai2010simultaneous,cai2009split,elad2005simultaneous,li2013adaptive,lim2013nonseparable}.

The (digital) image inpainting problem has been largely initiated in \cite{efros1999texture}.
Since then, many techniques and methods have been proposed in the literature to investigate the image inpainting problem, to mention only a few, for example, see
\cite{Bertalmio2014,Bertalmio2003,cai2008framelet,
cai2010simultaneous,cai2009split,chanshenshen,Chan2005,elad2005simultaneous,li2013adaptive,lim2013nonseparable} and many references therein.
According to the survey article \cite{Bertalmio2014}, methods for image inpainting can be generally classified into three different groups: patch-based methods, PDEs/variational methods,
and sparse representation based methods.
Most patch-based methods first take a reference patch which is a small subimage containing both missing pixels and known pixels of a given observed image. Then one or several similar patches to the reference patch are chosen from a neighboring region of the reference patch. Finally, the missing pixels inside the reference patch are filled with the information and pixels from these similar patches. After the missing pixels of a reference patch are filled/recovered, the patch-based methods move on to a new reference patch
(but never come back to previously visited reference patches again) and repeat this inpainting procedure. PDEs/variational methods employ either a variational principle through a minimization scheme or a partial differential equation (PDE) to fill the missing pixels inside an inpainting mask.
Sparse representation based methods are built on the assumption that natural images are often sparse under certain bases, frames, or dictionaries.
Interested readers can check the survey article \cite{Bertalmio2014} and the books \cite{Chan2005,elad:book}
for a detailed list of references on many different methods for the image inpainting problem.

The frame-based iterative image inpainting algorithms are more recent and belong to
the sparse representation based methods.
In this section we mainly review several related frame-based iterative inpainting algorithms
which have been proposed in \cite{cai2008framelet,cai2010simultaneous,cai2009split,elad2005simultaneous,li2013adaptive,lim2013nonseparable}.
All these frame-based iterative algorithms employ either a single frame or two tight frames with one for the cartoon component and the other for the texture component.
The soft-thresholding and hard-thresholding functions, which have been frequently used in many applications including most iterative frame-based image inpainting algorithms, are given by
\[
\eta^{soft}_{\lambda}(c)=
\begin{cases}
c-\lambda\tfrac{c}{|c|}, &|c|>\lambda,\\
0, &\text{otherwise}
\end{cases}
\qquad \mbox{and}\qquad
\eta^{hard}_{\lambda}(c)=
\begin{cases}
c, &|c|>\lambda,\\
0, &\text{otherwise},
\end{cases}
\]
where $c\in \C$ and $\lambda\ge 0$ is the thresholding value.
In this section, $\eta_\lambda$ will denote either the soft-thresholding function or the hard-thresholding function as well as their variants.

We first discuss inpainting algorithms employing a single tight frame $\ws$.
Let $\ws\in \R^{d\times n}$ be a given tight frame and $\{1, \ldots, d\}\backslash \Omega$ be a given inpainting mask (that is, data is only observable on the region $\Omega$). Recall that the $d\times d$ projection matrix $\cP$ is defined in \eqref{proj}.
To recover an unknown clean image $\x$ from the observed noisy image $\y$ for the image inpainting problem in \eqref{iip}, all the frame-based iterative inpainting algorithms in \cite{cai2008framelet,chanshenshen,li2013adaptive,lim2013nonseparable} and our proposed inpainting algorithm in this paper use a single tight frame $\ws$
and take the following form:
Let $\x_0=0$ be the initial guess and
\be\label{iterative}
\x_{\ell} = \cP\y + (I-\cP) (\ws \eta_{\Lambda_\ell}(\ws^\tp \x_{\ell-1})), \qquad \ell\in \N,
\ee
where $\Lambda_\ell:=(\lambda_{\ell,1}, \ldots, \lambda_{\ell,n})^\tp$ are vectors of nonnegative thresholding values and the thresholding operator $\eta_{\Lambda_\ell}$ is defined to be
\be
\eta_{\Lambda_\ell}((c_1,\ldots,c_n)^\tp) =(\eta_{\lambda_{\ell,1}}(c_1),\ldots,\eta_{\lambda_{\ell,n}}(c_n))^\tp
\ee
such that the sequence of thresholding values $\{\lambda_{\ell,j}\}_{\ell=1}^\infty$ decreases (often to $0$ or the noise level) as $\ell \to \infty$ for every $j=1, \ldots, n$.
The thresholding vectors $\Lambda_\ell$ are often chosen to be the form $r^\ell \Lambda$ for some $0<r<1$ and an initial thresholding vector $\Lambda:=(\lambda_1, \ldots, \lambda_n)^\tp$.
The iterative algorithm stops when the error $\|\x_{\ell}-\x_{\ell-1}\|_2$ is less than a given tolerance and the output $\x_\ell$ is regarded as a restored/recovered inpainted image of the unknown clean image $\x$ in \eqref{iip}.

We now recall tight framelet filter banks which have been used in several frame-based image inpainting algorithms.
A one-dimensional filter $a=\{a(k)\}_{k\in \Z}: \Z \rightarrow \C$ is simply a sequence on $\Z$ and we define $\wh{a}(\xi):=\sum_{k\in \Z} a(k) e^{-ik\xi}$, $\xi\in \R$ which is a $2\pi$-periodic function.
For filters $a, b_1, \ldots, b_s: \Z \rightarrow \C$, we say that $\{a; b_1, \ldots, b_s\}$ is a tight framelet filter bank if
\[
|\wh{a}(\xi)|^2+|\wh{b_1}(\xi)|^2+\cdots+|\wh{b_s}(\xi)|^2=1
\quad \mbox{and}\quad
\wh{a}(\xi)\overline{\wh{a}(\xi+\pi)}+\wh{b_1}(\xi)\overline{\wh{b_1}(\xi+\pi)}+
\cdots+\wh{b_s}(\xi)\overline{\wh{b_s}(\xi+\pi)}=0.
\]
For two-dimensional (2D) problems, one simply uses the tensor product 2D tight framelet filter bank:
\[
\{a; b_1, \ldots, b_s\}\otimes \{a; b_1, \ldots, b_s\}:=
\{a\otimes a; b_\ell\otimes b_m, \ell, m=1, \ldots, s\},
\]
where the tensor product filter $u\otimes v$ is defined to be $[u\otimes v](j,k):=u(j)v(k)$ for $j,k\in \Z$ for two one-dimensional filters $u,v: \Z \rightarrow \C$. See \cite{daub:book,dong2010mra,HanACHA1997,HanMMNP2013,RonShenJFA1997,shen2010wavelet} and many references therein for some background on tight framelets, wavelet filter banks, and wavelet analysis.

The inpainting algorithm in \cite{cai2008framelet,chanshenshen} uses the soft-thresholding function and the underlying tight frame $\ws$ is generated from the tensor product of the spline tight framelet filter bank $\{a; b_1, \ldots, b_4\}$ (see \cite{RonShenJFA1997}):
\begin{align*}
&a=\tfrac{1}{16}\{ 1, 4, 6, 4, 1\}_{[-2,2]},\qquad
 b_1=\tfrac{1}{8}\{ 1, 2, 0, -2, -1\}_{[-2,2]}, \quad b_2=\tfrac{\sqrt{6}}{16}\{-1, 0, 2, 0, -1\}_{[-2,2]}, \\
& b_3=\tfrac{1}{8}\{-1, 2, 0, -2, 1\}_{[-2,2]}, \quad  b_4=\tfrac{1}{16}\{ 1, -4, 6, -4, 1\}_{[-2,2]}.
\end{align*}
Since the tensor product tight framelet filter bank $\{a; b_1, \ldots, b_4\}\otimes \{a; b_1, \ldots, b_4\}$ has $25$ filters and is implemented in an undecimated fashion with one decomposition level, the redundancy rate of the tight frame $\ws$ used in
\cite{cai2008framelet,chanshenshen} is $25$, that is, $\ws\in \R^{d\times n}$ with $n=25d$.

The inpainting algorithm in \cite{li2013adaptive} uses a local soft-thresholding function (see \cite[Formula~(19)]{li2013adaptive} and Section~\ref{sec:tpctf} for more detail)
and the underlying tight frame $\ws$ is generated from the undecimated discrete cosine transform. More precisely, the matrix representation of the discrete cosine transform with a block size $m\times m$ is an orthonormal $m\times m$ matrix given by
\be\label{matrixC}
B: = \frac{1}{\sqrt{m}} \left[  \epsilon_j \cos \frac{(j-1)(2k-1)\pi}{2m} \; : \;  k,j = 1,2,\ldots, m  \right],
\ee
where $\epsilon_1=1$ and $\epsilon_j=\sqrt{2}$ for $j=2,\ldots,m$.
For $j=1, \ldots, m$, let $B_j$ be the $j$th column of the matrix $\frac{1}{\sqrt{m}} B$, that is, $B_j:=\{\frac{1}{m} \epsilon_j \cos\frac{(j-1)(2k-1)\pi}{2m}\}_{1\le k\le m}$ which is regarded as a sequence on $\Z$ with support $[1,\ldots, m]$. Then $\{B_1, \ldots, B_m\}$ forms an undecimated tight framelet filter bank, that is, $\sum_{j=1}^m |\wh{B_j}(\xi)|^2=1$.
The two-dimensional tensor product tight framelet filter bank $\{B_1, \ldots, B_m\}\otimes \{B_1, \ldots, B_m\}$ is used
to generate the tight frame $\ws$ in \cite{li2013adaptive}.
Since the block size $m=7$ in \cite{li2013adaptive}, there are $49$ filters in $\{B_1, \ldots, B_m\}\otimes \{B_1, \ldots, B_m\}$. Hence, the tight frame in \cite{li2013adaptive} has redundancy rate $49$, since all the tensor product filters are implemented in an undecimated fashion. We point out that \cite{li2013adaptive} uses an accelerated version of the iterative algorithm in \eqref{iterative}.

The inpainting algorithm in \cite{lim2013nonseparable} uses the hard-thresholding function instead
and its underlying frame $\ws$ is generated from compactly supported nonseparable shearlets having one two-dimensional low-pass filter and $16, 16, 8, 8$ two-dimensional high-pass filters at $4$ different resolutions.
Since there are total $49$ filters which have been implemented in an undecimated fashion, the frame used in \cite{lim2013nonseparable} has the redundancy rate $49$.
But the frame used in \cite{lim2013nonseparable} is not a tight frame and a (non-compactly supported) dual frame $\tilde{\ws}$ is used in \cite{lim2013nonseparable} to achieve perfect reconstruction.

The general iterative algorithm in \eqref{iterative} is closely linked to the balanced approach in \eqref{balanced}.
If all the thresholding vectors $\Lambda_\ell=\Lambda=(\lambda_1, \ldots, \lambda_n)^\tp$ are constants for all $\ell\in \N$ and if $\eta_\Lambda$ uses the soft-thresholding function,
then it has been proved in \cite{cai2008framelet} that the sequence $\{\mc_\ell := \eta^{soft}_{\Lambda}(\ws^\tp \x_\ell)\}_{\ell\in \N}$
converges to the minimizer $\breve{\mc}$ of the following convex minimization problem using
the balanced approach as follows:
\be\label{B1}
\min_{\mc\in\R^n} \frac{1}{2}\| \cP\ws \mc-\cP\y \|_2^2 + \|\diag (\Lambda)\mc \|_1+ \kappa \|(I - \ws^\tp \ws)\mc\|_2^2 \qquad \mbox{with}\quad \kappa=1/2.
\ee
The restored image is given by $\ws \hmc$.

We now discuss simultaneous cartoon and texture inpainting algorithms in \cite{cai2010simultaneous,cai2009split,elad2005simultaneous}.
Natural images $\x$ often have two layers of structures: the cartoon part $\x_c$ for geometry structures and the texture part $\x_t$ for oscillating patterns.
One adopts two tight frames which can represent sparsely the cartoon and texture parts of the image respectively.
Let $\ws_c$ be a tight frame for sparsely representing the cartoon part of an image, and let $\ws_t$ be a tight frame for sparely approximating the texture part of an image. A cartoon and texture inpainting algorithm called
(Morphological Component Analysis) MCA-based algorithm is introduced in \cite{elad2005simultaneous}. The MCA-based iterative algorithm is given as follows: Let $\x_{t,0}= 0$ and $\x_{c,0}=0$. For $\ell\in \N$,
\begin{align}
  &\x_{t,\ell} =  \ws \eta_{\Lambda_{t,\ell}} \left( \ws^{\tp}_t  (\cP \y + (I-\cP)\x_{t,\ell-1} - \cP \x_{c,\ell-1} ) \right), \nonumber \\
  &\x_{c,\ell} =  \ws \eta_{\Lambda_{c,\ell}} \left( \ws^{\tp}_c (\cP \y + (I-\cP)\x_{c,\ell-1} - \cP \x_{t,\ell-1} )\right), \nonumber \\
  &\x_{c,\ell} =  \x_{c,\ell} - \mu \nabla \cdot \left(\frac{\nabla \x_{c,\ell}}{|\nabla \x_{c,\ell}|}\right),   \nonumber \\
  &\mbox{update the vectors}\ \Lambda_{t,\ell}\, \mbox{and}\;  \Lambda_{c,\ell}\ \mbox{of thresholding values}, \label{MCA}
\end{align}
where $\mu$ is a constant and $\nabla$ is the gradient operator. The output $\x_{t,\ell}+\x_{c,\ell}$ is regarded as the  restored/recovered inpainted image. If $\eta_{\Lambda_{t,\ell}}$ and $\eta_{\Lambda_{c,\ell}}$ use the soft-thresholding function, as argued by a heuristic argument in \cite{elad2005simultaneous}, the above MCA-based algorithm in
\cite{elad2005simultaneous} is roughly a block-coordinate-relaxation
method to solve the following model approximately:
\be\label{eqd1d2}
\min_{\x_t,\x_c \in \R^d} \frac{1}{2}\| \cP(\x_t+\x_c-\y) \|^2_2 + \lambda_1\|\ws^{\tp}_t \x_t\|_1 + \lambda_2 \| \ws^{\tp}_c \x_c \|_1 + \lambda_3 \| \x_c \|_{TV},
\ee
where $\|\cdot\|_{TV}$ is the total variation and the three nonnegative regularization parameters $\lambda_1, \lambda_2, \lambda_3$ (often with $\lambda_1=\lambda_2$) depend on the choices of $\Lambda_{t,\ell}, \Lambda_{c,\ell}, \mu$ in the MCA-based algorithm.
To compare the MCA-based algorithm in \cite{elad2005simultaneous} with our proposed algorithm, we choose the tight frame $\ws_c$ built from the curvelet transform with five decomposition levels and the tight frame $\ws_t$ built from the
$50\%$ overlapping discrete cosine transform with a block size $32\times 32$. The tight frame $\ws_t$ has redundancy rate $2.89$ and
 $\ws_t$ has redundancy rate $2$.
Thus,
the total redundancy rate of the frames is $4.89$.

For $\lambda_3=0$ in \eqref{eqd1d2}, \cite{cai2009split} proposes the split Bregman method to solve the model \eqref{eqd1d2} with given regularization parameters $\lambda_1$ and $\lambda_2$.
Instead of updating the residuals in the image domain, the following iterative algorithm, proposed in \cite{cai2010simultaneous}, updates the residuals in the frame transform domain as follows: For $\ell\in \N$,
\begin{align}
  &\mc_{t,\ell} = \eta_{\Lambda_{t,\ell}}\left((1-\kappa_t \gamma)\mc_{t,\ell-1}+\gamma \ws_t^{\tp}(\kappa_t \ws_t \mc_{t,\ell-1} + \cP(\y-\ws_t\mc_{t,\ell-1} -\ws_c \mc_{c,\ell-1})) \right), \nonumber\\
  &\mc_{c,\ell} = \eta_{\Lambda_{c,\ell}}\left((1-\kappa_c \gamma)\mc_{c,\ell-1}+\gamma \ws_c^{\tp}(\kappa_c \ws_t \mc_{t,\ell-1} + \cP(\y-\ws_t\mc_{t,\ell-1} -\ws_c \mc_{c,\ell-1})) \right), \nonumber\\
  &\mbox{If necessary, update the vectors}\ \Lambda_{t,\ell}\; \mbox{and}\;   \Lambda_{c,\ell}\ \mbox{of thresholding values}, \label{Cai}
\end{align}
where $\kappa_t$, $\kappa_c$ and $\gamma$ are positive constants. The parameter $\gamma$ is usually less than $\frac{1}{\max\{1,\kappa_t,\kappa_c\}}$. The output $\ws_t\mc_{t,\ell}+\ws_c\mc_{c,\ell}$ is regarded as the restored/recovered image.
Fixing the thresholding values, the iterative algorithm \eqref{MCA} without the total variation term (that is, $\lambda_3=0$) and the iterative algorithm \eqref{Cai} look very similar to each other. It has been proved in \cite{cai2010simultaneous} that the sequences $\{\mc_{t,\ell}\}_{\ell\in \N}$ and $\{\mc_{c,\ell}\}_{\ell\in \N}$, generated by the iterative algorithm \eqref{Cai} with fixed thresholding vectors
$\Lambda_{t,\ell}=\Lambda_t$ and $\Lambda_{c,\ell}=\Lambda_c$ for all $\ell\in \N$,
converge to the solution of the following balanced approach using two tight frames:
\be\label{model4}
\begin{split}
\min_{\mc_t,\mc_c\in\R^n} &\frac{1}{2}\Big \|\cP(\ws_t \mc_{t}+\ws_c \mc_c-\y) \Big \|_2^2
+ \|\diag (\Lambda_t)\mc_{t} \|_1+\|\diag (\Lambda_c)\mc_{c} \|_1\\
&\qquad\qquad
+ \frac{\kappa_t}{2} \|(I - \ws_t^\tp \ws_t)\mc_t\|_2^2+ \frac{\kappa_c}{2} \|(I - \ws_c^\tp \ws_c)\mc_c\|_2^2.
\end{split}
\ee
For the numerical experiments in \cite{cai2010simultaneous,cai2009split}, the tight frame $\ws_c$ is built from the tensor product of the following spline tight framelet filter bank $\{a; b_1, b_2\}$ (see \cite{RonShenJFA1997}):
\[
a=\tfrac{1}{4}\{ 1, 2, 1\}_{[-1,1]}, \quad
b_1=\tfrac{\sqrt{2}}{4}\{ -1, 0, 1\}_{[-1,1]}, \quad
b_2=\tfrac{1}{4}\{-1,  2,  -1\}_{[-1,1]}.
\]
Since there are total $9$ filters in $\{a; b_1, b_2\}\otimes \{a; b_1, b_2\}$ and all the filters are implemented in an undecimated fashion in \cite{cai2010simultaneous,cai2009split} with one decomposition level, the redundancy rate of $\ws_c$ is $9$.
The tight frame $\ws_t$ is built from the  $50\%$ overlapping discrete cosine transform with a block size $16\times 16$. The tight frame $\ws_t$ has the redundancy rate $2$ and hence, the total redundancy rate is $11$.

\section{Directional Tensor Product Complex Tight Framelets}
\label{sec:tpctf}

Since we shall use directional tensor product complex tight framelets ($\tpctf$s) in our image inpainting algorithm, in this section we briefly review $\tpctf$s which have been initially introduced in \cite{HanMMNP2013} and further developed in \cite{HanZhaoSIIMS2014} for the image denoising problem.

For $a=\{a(k)\}_{k\in \Z} \in l_1(\Z)$, recall that $\wh{a}(\xi):=\sum_{k\in \Z} a(k)e^{-i k\xi}$ as its Fourier series which is a $2\pi$-periodic function and $a(k)=\frac{1}{2\pi}\int_{-\pi}^\pi \wh{a}(\xi) e^{ik\xi} d\xi$ for all $k\in \Z$.
Let us first discuss how to modify a characteristic function of an interval into a smooth bump function (\cite{daub:book,HanACHA1997}). Let $\pP_{m}(x):=(1-x)^m \sum_{j=0}^{m-1} \binom{m+j-1}{j} x^j$. Then $\pP_{m}$ satisfies
the identity $\pP_{m}(x)+\pP_{m}(1-x)=1$ (see \cite{daub:book}). For $c_L<c_R$ and two positive numbers $\gep_L, \gep_R$ satisfying $\gep_L+\gep_R\le c_R-c_L$, we define a bump function $\chi_{[c_L, c_R]; \gep_L, \gep_R}$ on $\R$ by
\be \label{bump:func}
\chi_{[c_L, c_R]; \gep_L, \gep_R}(\xi):=
\begin{cases} 0, \quad &\xi\le c_L-\gep_L \; \mbox{or}\; \xi \ge c_R+\gep_R, \\
\sin\big(\tfrac{\pi}{2}\pP_{m}(\tfrac{c_L+\gep_L-\xi}{2\gep_L})\big), \quad &c_L-\gep_L<\xi<c_L+\gep_L,\\
1, \quad &c_L+\gep_L\le \xi\le c_R-\gep_R,\\
\sin\big(\tfrac{\pi}{2}\pP_{m}(\tfrac{\xi-c_R+\gep_R}{2\gep_R})\big),
\quad &c_R-\gep_R<\xi<c_R+\gep_R.
\end{cases}
\ee
Let $s\in \N$ be a positive integer. Let $c_1, \gep_1$ be positive real numbers satisfying
\[
0<\gep_1\le \min(c_1, \tfrac{\pi}{2}-c_1, \tfrac{c_1+(s-1)\pi}{2s}).
\]
A low-pass filter $a$ and high-pass filters $b^{1, p}, \ldots, b^{s,p}, b^{1,n}, \ldots, b^{s,n}$ have been constructed in \cite{HanMMNP2013,HanZhaoSIIMS2014}
by defining their Fourier series on the basic interval $[-\pi, \pi)$ as follows:
\be \label{ctf}
\wh{a}:=\chi_{[-c_1, c_1]; \gep_1, \gep_1}, \quad
\wh{b^{\ell,p}}:=\chi_{[c_\ell,c_{\ell+1}]; \gep_1, \gep_1}, \quad
\wh{b^{\ell,n}}:=\ol{\wh{b^{\ell,p}(-\cdot)}}, \qquad \ell=1, \ldots, s,
\ee
where $c_\ell:=c_1+\frac{\pi-c_1}{s}(\ell-1)$ for $\ell=1, \ldots, s$.
Then it is straightforward to directly check that $\ctf_{2s+1}:=\{a; b^{1,p}, \ldots, b^{s,p}, b^{1,n}, \ldots, b^{s,n}\}$ forms a tight framelet filter bank, that is, the following identities hold:
\be \label{tffb:1}
|\wh{a}(\xi)|^2+\sum_{\ell=1}^s |\wh{b^{\ell,p}}(\xi)|^2+\sum_{m=1}^s |\wh{b^{m,n}}(\xi)|^2=1,
\qquad a.e.\; \xi\in [-\pi, \pi],
\ee
\be \label{tffb:0}
\begin{split}
\wh{a}(\xi)\ol{\wh{a}(\xi+\pi)}+\sum_{\ell=1}^s \wh{b^{\ell,p}}(\xi)\ol{\wh{b^{\ell,p}}(\xi+\pi)}&+\sum_{m=1}^s \wh{b^{m,n}}(\xi)
\ol{\wh{b^{m,n}}(\xi+\pi)}=0,
\qquad a.e.\; \xi\in [-\pi, \pi].
\end{split}
\ee
The tensor product complex tight framelet filter bank $\tpctf_{2s+1}$ for dimension two is simply
\begin{align*}
\tpctf_{2s+1}&:=\ctf_{2s+1}\otimes \ctf_{2s+1}\\
&=\{a; b^{1,p}, \ldots, b^{s,p}, b^{1,n}, \ldots, b^{s,n}\}\otimes \{a; b^{1,p}, \ldots, b^{s,p}, b^{1,n}, \ldots, b^{s,n}\}.
\end{align*}

To further improve the directionality of $\tpctf_{2s+1}$, another closely related family of tensor product complex tight framelet filter banks $\tpctf_{2s+2}$ has been introduced in \cite{HanZhaoSIIMS2014}. Let $0<\gep_0<c_1-\gep_1$. Define filters $a, b^{1,p}, \ldots, b^{s,p}, b^{1,n}, \ldots, b^{s,n}$ as in \eqref{ctf} and define two auxiliary low-pass filters $a^p, a^n$ by
\be \label{tpctf:even}
\wh{a^{p}}:=\chi_{[0, c_1]; \gep_0, \gep_1}, \qquad
\wh{a^{n}}:=\ol{\wh{a^p}(-\cdot)}.
\ee
We can directly check that $\{a^p, a^n; b^{1,p}, \ldots, b^{s,p}, b^{1,n}, \ldots, b^{s,n}\}$ is still a tight framelet filter bank. Now the tensor product complex tight framelet filter bank $\tpctf_{2s+2}$ for dimension two is defined to be
\[
\tpctf_{2s+2}:=\{a\otimes a; \tpctf\mbox{-HP}_{2s+2}\},
\]
where $\tpctf\mbox{-HP}_{2s+2}$ consists of total
$4s(s+2)$ high-pass filters given by
\[
a^p\otimes b^{\ell,p}, a^p\otimes b^{\ell,n}, a^n\otimes b^{\ell,p}, a^n\otimes b^{\ell,n},
b^{\ell,p}\otimes b^{m,p}, b^{\ell,p}\otimes b^{m,n},
b^{\ell,n}\otimes b^{m,p}, b^{\ell,n}\otimes b^{m,n},\qquad
\ell,m=1, \ldots, s.
\]
See \cite{HanMMNP2013,HanZhaoSIIMS2014}
for more details and explanation on directional tensor product complex tight framelets.

For our image inpainting algorithm, we shall use the particular tensor product complex tight framelet filter bank $\tpctf_6$ with $c_1=\frac{119}{128}$, $\gep_0=\frac{35}{128}$ and $\gep_1=\frac{81}{128}$. This tensor product complex tight framelet filter bank $\tpctf_6$ has been used in \cite{HanZhaoSIIMS2014} for image denoising with impressive performance.
To illustrate the directionality of this $\tpctf_6$, the elements in $\tpctf_6$ for dimension two are given in Figure~\ref{tpctf6graph} (also see \cite[Figure~8]{HanZhaoSIIMS2014}).

\begin{figure}[ht]
\centerline{
\includegraphics[width=5.0in,height=2.0in]{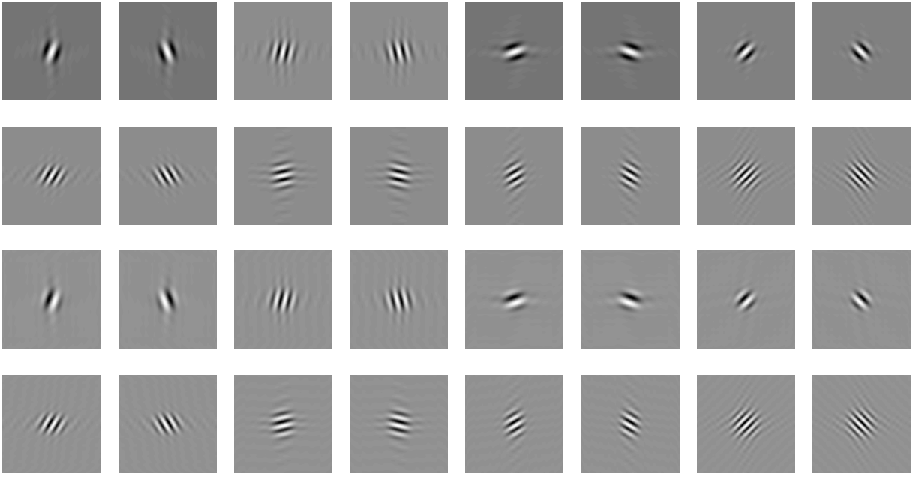}}
\begin{caption}{
The first two rows show the real part and the last two rows show the imaginary part of the elements in $\tpctf_6$.
Among these 16 graphs, the directions along $\pm 45^\circ$ are repeated once. Hence, there are total 14 directions in $\tpctf_6$ for dimension two.
} \label{tpctf6graph}
\end{caption}
\end{figure}

For $\x=(x_1, \ldots, x_d)^\tp\in \R^d$ (The integer $d$ is often a power of $2$),
we can uniquely extend $\x$ into a $d$-periodic sequence $\{x_k\}_{k\in \Z}$ on $\Z$ such that $x_{d+k}=x_k$ for all $k\in \Z$. Note that there are $32$ real-valued two-dimensional high-pass filters in $\tpctf_6$ by separating the real and imaginary parts of each filter.
Therefore, the multilevel discrete framelet transform (see \cite{HanMMNP2013}) employing the tensor product complex tight framelet filter bank $\tpctf_6$ induces a real-valued tight frame $\ws\in \R^{d\times n}$ with
$n\le 32 d\sum_{j=1}^\infty \frac{1}{4^j}=\frac{32}{3}d$.
Consequently, regardless of the decomposition level, the redundancy rate of the tight frame built from two-dimensional $\tpctf_6$ is no more than $\frac{32}{3}\approx 10.67$.

To take the advantage of the cross-scale relations in the wavelet tree of frame coefficients,
instead of using the soft thresholding or hard thresholding,
we shall employ the bivariate shrinkage $\eta_\lambda^{bs}$ function which has been introduced in \cite{sendur2002bivariate} as follows:
\be \label{bs}
\eta^{bs}_\lambda(c)=\eta_{\lambda_c}^{soft}(c)
=\begin{cases}
c-\lambda_c \tfrac{c}{|c|}, &|c|> \lambda_c,\\
0, &\text{otherwise},
\end{cases}
\qquad \mbox{with}\quad
\lambda_c:=\frac{\sqrt{3}\sigma_n^2}{\sigma_c \sqrt{1+|c_p/c|^2}},
\ee
where $\sigma_n:=\lambda \|b\|_2$ with $b$ being the high-pass filter inducing the frame coefficient $c$,
the frame coefficient $c_p$ is the parent coefficient of $c$ in the immediate higher scale, and
\[
\sigma_c:=\begin{cases} \sqrt{\breve{\sigma}_c^2-\sigma_n^2}, &\breve{\sigma}_c> \sigma_n,\\
0, &\text{otherwise}
\end{cases}
\qquad \mbox{with}\qquad
\breve{\sigma}_c^2:=\frac{1}{\#N_c}\sum_{j\in N_c} |c_j|^2,
\]
where $\#N_c$ is the cardinality of the set $N_c$ which is the $[-3,3]^2$ window centering around the frame coefficient $c$ at the band induced by the filter $b$. We point out that a similar local soft-thresholding strategy as in \eqref{bs} is used in \cite[Formula (19)]{li2013adaptive} with $\lambda_c=\frac{\sqrt{2}\sigma_n^2}{\sigma_c}$, where $\sigma_n=\sigma/7$ and
\[
\sigma_c=\max\Big\{ \sqrt{\Big(\tfrac{1}{\# N_c} \sum_{j\in N_c} \sqrt{2} |c_j|\Big)^2-\sigma_n^2}, 10^{-3}\Big\},
\]
where $N_c$ is the $[-4,4]^2$ window centering around the frame coefficient $c$ at the same band.

From Figure~\ref{tpctf6graph},
we see that both the real part and imaginary part of $\tpctf_6$ have six edge-like directional elements and ten texture-like (or DCT-like) directional elements.
As shown in \cite{HanZhaoSIIMS2014}, the algorithm using $\tpctf_6$ and bivariate shrinkage in \eqref{bs} performs superior for image denoising.
Inspired by the results in \cite{HanZhaoSIIMS2014} for image denoising, we shall apply the tensor product complex tight framelet filter bank $\tpctf_6$ in our image inpainting algorithm and hope that both the cartoon and texture parts of an image can be simultaneously handled well by $\tpctf_6$.
See Section~\ref{sec:experiment} for the performance of $\tpctf_6$ for the image inpainting problem.

\section{Proposed Image Inpainting Algorithm Using $\tpctf_6$ and Numerical Experiments}
\label{sec:experiment}

Using directional tensor product complex tight framelets, in this section we shall provide details for our proposed image inpainting algorithm which has been outlined in Section~\ref{sec:introduction}. Then we shall present some numerical experiments on our proposed inpainting algorithm and its performance compared with six well-known frame-based image inpainting algorithms in \cite{cai2008framelet,cai2010simultaneous,cai2009split,elad2005simultaneous,li2013adaptive,lim2013nonseparable}.

\subsection{The Proposed Image Inpainting Algorithm Using $\tpctf_6$}
We now provide full details in Algorithm~\ref{mainalg} for our proposed image inpainting algorithm.
Recall that the image inpainting problem is given in \eqref{iip}, where $\Omega$ is the given observable region (i.e., $\{1, \ldots, d\}\backslash \Omega$ is the given inpainting mask), $\y=(y_1, \ldots, y_d)^\tp$ is the observed partial image on the observable region $\Omega$, $n_j$ denotes the i.i.d. zero-mean Gaussian noise with standard deviation $\sigma$, and $\x$ is the unknown clean image to be restored/recovered.
Let $\ws$ be the tight frame built from the tensor product complex tight framelet filter bank $\tpctf_6$ as discussed in Section~\ref{sec:tpctf}. Our proposed iterative image inpainting algorithm using $\tpctf_6$ is given in Algorithm~\ref{mainalg}.

\begin{algorithm}
\caption{(Frame-based Image Inpainting Algorithm Using $\tpctf_6$)}\label{mainalg}
\begin{algorithmic}[1]
\REQUIRE The tight frame $\ws\in \R^{d\times n}$ built from $\tpctf_6$, an inpainting mask $\{1, \ldots, d\}\backslash \Omega$ (i.e., $\Omega$ is a given observable region), standard deviation $\sigma$ of i.i.d. zero-mean Gaussian noise, and an observed partial image $\y\in \R^d$ on the observable region $\Omega$.
 \STATE Initialization: $\x_{-1} =\x_0=0$, $i=\ell=1$, $r=1-\frac{\# \Omega}{n}$, where $\# \Omega$ denotes the cardinality of $\Omega$.
 \STATE Generate thresholding values $\Lambda_1, \Lambda_2$ by \eqref{Lambda} and iteration parameters $N_1, N_2, \tol_1, \tol_2$ by Table~\ref{tab:parameters}.
 \STATE Initialize thresholding value: $\lambda = \Lambda_1(1)$.
 \WHILE {($i\le N_1+N_2$)}
 \STATE  $\y_\ell = \cP(\y) + (I-\cP)(\x_\ell)$.
 \STATE  $\mc_{\ell+1} = \eta^{bs}_{\lambda} (\ws^\tp\y_\ell)$.
 \STATE  $\x_{\ell+1} = \ws\mc_{\ell+1}$.
 \STATE  $\error = \| (I-\cP)(\x_{\ell+1} - \x_{\ell}) \|_2/
 \| \cP \y\|_2$.
 \IF    {($\error < \tol_1$) \AND  ($i<N_1$)}
 \STATE  $i=i+1$.
 \STATE  $\lambda = \Lambda_1(i)$.
 \ELSIF  {($\error < \tol_2$) \AND ($N_1 \le i < N_1 +  N_2$)}
 \STATE  $i=i+1$.
 \STATE  $\lambda = \Lambda_2(i-N_1)$.
 \ELSIF  {($\error < \tol_2$) \AND  ($i = N_1 +  N_2$)}
 \STATE  Break.
 \ENDIF
 \STATE  $\ell=\ell+1$.
 \ENDWHILE
 \RETURN $\x_{\ell+1}$ as the restored image.
 \end{algorithmic}
 \end{algorithm}

The output $\x_{\ell+1}$ from Algorithm~\ref{mainalg} is the restored/inpainted image.
The projection operator $\cP$ is defined in \eqref{proj} and the thresholding $\eta^{bs}_\lambda$ is defined in \eqref{bs}.
Note that $r$ is the percentage of missing pixels, that is, the ratio between the number of missing pixels in the inpainting mask and the number of all pixels in an image.
We now discuss how to generate the thresholding values $\Lambda_1$ and $\Lambda_2$.
Algorithm~\ref{mainalg} uses decreasing thresholding values $\Lambda_1\cup \Lambda_2$ from $[\lambda_{\min}, \lambda_{\midd}]\cup [\lambda_{\midd}, \lambda_{\max}]$,
where we set
\be \label{lambdaminmax}
\lambda_{\min}:=\max\{1,\sigma(1-\tfrac{r^2}{2})\}, \qquad \lambda_{\max}:=512,
\qquad
\lambda_{\midd}:=\min \{\max \{2 \lambda_{\min}+10, 20\}, \lambda_{\max}\}.
\ee

The sequence $\Lambda_1$ of decreasing thresholding values on $[\lambda_{\midd}, \lambda_{\max}]$ and
the sequence $\Lambda_2$ of decreasing thresholding values on $[\lambda_{\min},\lambda_{\midd}]$ are given by
\be\label{Lambda}
\Lambda_1(i) =  r_1^{\frac{i-N_1}{N_1-1}} \lambda_{\midd},\quad i = 1,\ldots,N_1
\qquad \mbox{and}\qquad
\Lambda_2(i) =  r_2^{\frac{i-N_2}{N_2}} \lambda_{\min},\quad i = 1,\ldots,N_2,
\ee
where $r_1 := \frac{\lambda_{\midd} }{\lambda_{\max}}$ and
$r_2 := \frac{ \lambda_{\min} }{\lambda_{\midd}}$. Note that $0<r_1, r_2<1$ and
\[
\Lambda_1(1)=\lambda_{\max},\quad \Lambda_1(N_1)=\lambda_{\midd},\quad \Lambda_2(1)=r_2^{\frac{1}{N_2}} \lambda_{\midd},\quad \Lambda_2(N_2)=\lambda_{\min}.
\]

\begin{table}[htbp]
  \centering
  \caption{Choices of iteration parameters $N_1, N_2, \tol_1, \tol_2$ for Algorithm~\ref{mainalg}.} \label{tab:parameters}%
    \begin{tabular}{|c||c|c|c|c|}
    \hline
    $r$ & $N_1$ & $\tol_1$ & $N_2$ & $\tol_2$ \\
     \hline
    $0<r<0.5$ & $5$ & $ 5\times 10^{-3}$ & $8$ & $ 10^{-4}$ \\
     \hline
    $0.5\le r<1$ & $8$ & $ 5\times  10^{-3}$ & $5$ & $ 10^{-3}$ \\
     \hline
    \end{tabular}%
\end{table}%

We now provide some explanation for the choices of the parameters in Algorithm~\ref{mainalg}.
Note that $1-r=\frac{\# \Omega}{n}$ is the percentage of the observable region $\Omega$ over the total size of an image. Since the noise standard deviation is $\sigma$,
the total squared noise energy on the region $\Omega$ is $\sigma^2 (1-r)n$. Therefore, after the tight frame transform, the noise standard deviation of each frame coefficient (after normalization of the filters) is approximately $\sigma \sqrt{1-r}$. On the other hand, the boundary of the inpainting mask often creates artificial jumps between the observable region and missing region. Hence, the boundary of the inpainting mask also contributes as another source of noise. As a consequence, we should set the positive minimal thresholding value $\lambda_{\min}$ greater than $\sigma \sqrt{1-r}$.
For simplicity, we take $\lambda_{\min}$ as in \eqref{lambdaminmax} and one can directly check that $\lambda_{\min}\ge \sigma \sqrt{1-r}$. The maximal thresholding value $\lambda_{\max}$ can be any large enough number. Since greyscale images use the greyscale level $[0,255]$, it is natural for us to simply set $\lambda_{\max}=512$, although other large numbers can be used for $\lambda_{\max}$. The intermediate thresholding value $\lambda_{\midd}$ is chosen to be twice as $\lambda_{\min}$ while avoiding to be too small or too large.

When Algorithm~\ref{mainalg} iterates with decreasing thresholding values $\lambda$ from the first sequence $\Lambda_1\subseteq [\lambda_{\midd},\lambda_{\max}]$, the iterative scheme tends to recover the main geometry structures and edges (i.e., the cartoon part) of the corrupted image.
When Algorithm~\ref{mainalg} iterates with decreasing thresholding values $\lambda$ from the second sequence $\Lambda_2\subseteq [\lambda_{\min},\lambda_{\midd}]$, the iterative scheme focuses on restoring the fine detail structures (roughly speaking, textures) of the corrupted image. See Figure~\ref{fig:twostage} for an illustration.
A similar strategy of using deceasing sequences of thresholding values can be found in all
frame-based iterative image restoration algorithms in \cite{cai2008framelet,cai2010simultaneous,cai2009split,chanshenshen,elad2005simultaneous,li2013adaptive,lim2013nonseparable}.
Note that $r$ is the percentage of the missing part (i.e., inpainting mask) over the total image size. Thus, if $r$ is large (i.e., a large portion of pixels is missing with large `holes'), then we can only expect to recover the geometry/background structures and hence, a large iteration step $N_1$ and a small iteration step $N_2$ are reasonable.
If $r$ is small (a small portion of pixels is missing with small `holes'), then we should emphasize on restoring fine structures by using a small iteration step $N_1$ and a large iteration step $N_2$ with a high precision $\tol_2$.

\begin{figure}[htb]
\centering
\subfigure[Clean image]{\includegraphics[width=3cm,height=2.6cm]{9.png}}
\subfigure[Corrupted image]{\includegraphics[width=3cm,height=2.6cm]{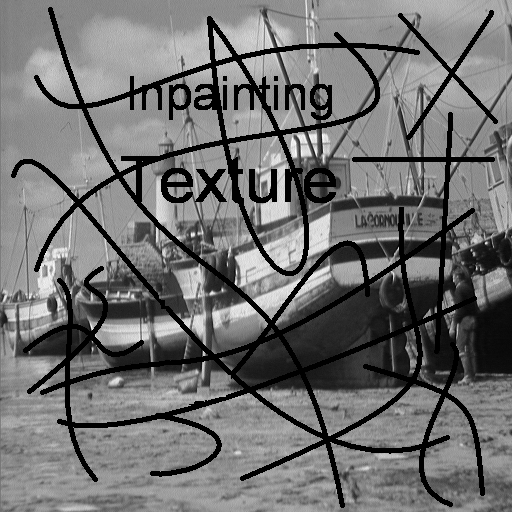}}
\subfigure[Restored (first stage)]{\includegraphics[width=3cm,height=2.6cm]{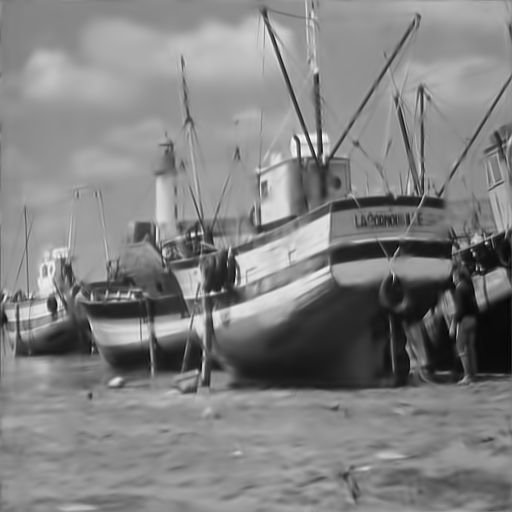}}
\subfigure[Restored (second stage)]{\includegraphics[width=3cm,height=2.6cm]{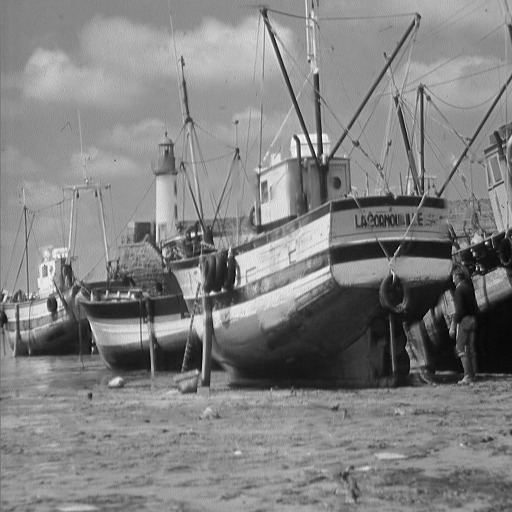}}
\caption{(a) is the $512\times 512$ clean image Boat.
(b) is the corrupted image by the inpainting mask Text~3 in Figure~\ref{testimages} without noise.
(c) and (d) are restored images by our proposed Algorithm~\ref{mainalg} with two decreasing sequences $\Lambda_1$ and $\Lambda_2$ of thresholding values given in \eqref{Lambda}.}\label{fig:twostage}
\end{figure}

\subsection{Numerical Experiments on Image Inpainting without Noise}

We now provide numerical experiments to test performance of our proposed Algorithm~\ref{mainalg} using tensor product complex tight framelet filter bank $\tpctf_6$ and compare its performance with six well-known frame-based image inpainting algorithms in
\cite{cai2008framelet,cai2010simultaneous,cai2009split,elad2005simultaneous,li2013adaptive,lim2013nonseparable}.
Recall that \cite{cai2008framelet} uses undecimated real-valued spline tight framelets, \cite{li2013adaptive} uses undecimated  DCT-Haar wavelet filters, and \cite{lim2013nonseparable} uses undecimated compactly supported nonseparable shearlets.
On the other hand, \cite{cai2010simultaneous,cai2009split,elad2005simultaneous} adopt simultaneous cartoon and texture inpainting by using two tight frames $\ws_c$ and $\ws_t$. \cite{elad2005simultaneous} uses $\ws_c$ built from curvelets for cartoon and $\ws_t$ built from the $50\%$ overlapping DCT for texture, while \cite{cai2010simultaneous,cai2009split} uses $\ws_t$ built from undecimated real-valued spline tight framelets for cartoon and $\ws_t$ built from the $50\%$ overlapping DCT for texture. See Section~\ref{sec:related} for more details on these frame-based iterative inpainting algorithms in \cite{cai2008framelet,cai2010simultaneous,cai2009split,elad2005simultaneous,li2013adaptive,lim2013nonseparable}.
The implementations of all the algorithms in \cite{cai2008framelet,cai2010simultaneous,cai2009split,elad2005simultaneous,li2013adaptive,lim2013nonseparable}
are provided by their own authors or downloaded from their homepages. We run all the related image inpainting algorithms with their default parameter values in their available codes. The greyscale test images and inpainting masks are given in Figure~\ref{testimages}.
Let $\x\in\R^d$ denote the true clean image and $\breve{\x}\in\R^d$ denote a restored/inpainted image. The quality of the restored/inpainted image is measured as usual by the Peak Signal-to-Noise Ratio (PSNR) defined by:
$\text{PSNR} = 10 \log_{10} (\tfrac{ 255^2 d }{\|\x-\breve{\x}\|_2^2})$.
The larger the PSNR value is, the better performance the inpainting algorithm is.

We now report the performance and comparison of our proposed Algorithm~\ref{mainalg} for the image inpainting problem without noise (that is, $\sigma=0$).
Since the performance and comparison of Algorithm~\ref{mainalg} for five $256\times 256$ greyscale test images in Figure~\ref{testimages} have been reported in Tables~\ref{tab:nonoise}--\ref{tab:noise}, we only report in Tables~\ref{tab:clean} and~\ref{tab:5122} the performance and comparison for five $512\times 512$ greyscale test images with $512\times 512$
inpainting masks Text~3 and Text~4 in Figure~\ref{testimages}, as well as with randomly missing pixels.

\begin{table}[htbp]
\centering
\caption{Performance in terms of PSNR values of several image inpainting algorithms without noise for five $512\times 512$ test images in (f)-(j) of Figure~\ref{testimages}.
Top left table for inpainting mask Text~3 in (m) of Figure~\ref{testimages}.
Top right table for inpainting mask Text~4 in (n) of Figure~\ref{testimages}.
Bottom left table for $50\%$ randomly missing pixels.
Bottom right table for $80\%$ randomly missing pixels.
Gain refers to the PSNR gain of our proposed Algorithm~\ref{mainalg} in row 7 over the highest PSNR values (underlined) among all the PSNR values in
the first 6 rows obtained by the $6$ inpainting algorithms in \cite{cai2008framelet,cai2010simultaneous,cai2009split,elad2005simultaneous,li2013adaptive,lim2013nonseparable}.
} \label{tab:clean}
    \begin{tabular}{|c|c|c|c|c|c||c|c|c|c|c|c|}
    \hline
      & Hill & Man & Boat & Barbara & Mandrill & Hill & Man & Boat & Barbara & Mandrill \\
      \hline
       \cite{cai2008framelet} & 35.08  & 34.50  & 33.10  & 31.89  & 29.26  & 31.63  & 30.70  & 29.29  & 29.04  & 26.84  \\
    \cite{cai2010simultaneous} & 34.94  & 33.99  & 33.05  & 34.52  & 28.85  & 31.71  & 30.36  & 29.35  & 30.88  & 26.47  \\
      \cite{cai2009split} & 34.00  & 33.34  & 32.93  & 34.11  & 28.82  & 31.01  & 29.55  & 28.84  & 30.56  & 26.30  \\
        \cite{elad2005simultaneous} & 34.98  & 33.68  & 33.77  & 34.24  & 28.91  & 31.76  & 30.28  & 29.85  & 31.17  & 26.30  \\
    \cite{li2013adaptive}  &\underline{35.73}  & 34.82  & 34.62  & 35.03  &\underline{29.89}  & \underline{\textbf{32.85}} & 31.34  & 30.35  & 31.51  &\underline{27.14}  \\
     \cite{lim2013nonseparable} & 35.69  &\underline{35.11}  &\underline{34.66}  &\underline{35.17}  & 29.66  & 32.13  &\underline{31.52}  &\underline{30.65}  &\underline{32.45}  & 27.02  \\
    Alg. \ref{mainalg} & \textbf{36.03} & \textbf{35.52} & \textbf{34.94} & \textbf{36.59} & \textbf{30.23} & 32.54  & \textbf{31.92} & \textbf{30.76} & \textbf{32.65} & \textbf{27.60} \\
    \hline
    Gain & 0.30  & 0.41  & 0.28  & 1.42  & 0.34  & -0.31 & 0.40  & 0.11  & 0.20  & 0.46  \\
    \hline\hline
     \cite{cai2008framelet} & 28.90  & 28.18  & 27.02  & 24.32  & 25.34  & 28.93  & 28.06  & 27.03  & 24.32  & 21.78  \\
    \cite{cai2010simultaneous} & 33.16  & 32.14  & 32.12  & 32.71  & 24.79  & 28.68  & 27.66  & 26.70  & 26.60  & 21.52  \\
      \cite{cai2009split} & 28.19  & 27.15  & 26.40  & 26.51  & 24.73  & 28.16  & 27.01  & 26.42  & 26.40  & 21.46  \\
       \cite{elad2005simultaneous} & 32.79  & 32.09  & 31.64  & 31.72  & 25.02  & 28.18  & 27.44  & 26.82  & 25.30  & 21.53  \\
    \cite{li2013adaptive}  &\underline{34.43}  &\underline{33.45}  &\underline{34.08}  & 33.85  &\underline{26.09}  & 28.99  & 27.69  & 27.87  & 26.39  & 21.05  \\
     \cite{lim2013nonseparable} & 33.11  & 32.81  & 33.07  &\underline{34.13}  & 25.33  & \underline{29.07}  &\underline{28.42}  &\underline{28.01}  &\underline{28.08}  &\underline{21.80}  \\
   Alg. \ref{mainalg} & \textbf{34.53} & \textbf{34.25} & \textbf{34.42} & \textbf{35.69} & \textbf{26.52} & \textbf{29.59} & \textbf{29.15} & \textbf{28.56} & \textbf{28.11} & \textbf{22.28} \\
   \hline
    Gain & 0.10  & 0.80  & 0.34  & 1.56  & 0.43  & 0.52 & 0.73  & 0.55  & 0.03  & 0.48  \\
    \hline
    \end{tabular}%
\end{table}%

\begin{table}[htbp]
  \centering
  \caption{Each (averaged) PSNR value here is the average of the PSNR values in Table~\ref{tab:clean} for the five test images for each inpainting algorithm. (Averaged) Gain here refers to the gain of the averaged PSNR value of our Algorithm~\ref{mainalg} over the averaged PSNR value of the corresponding inpainting algorithm.}  \label{tab:5122}
    \begin{tabular}{|c|cc|cc|cc|cc|}
    \hline
      & Text 3  & Gain & Text 4 & Gain & 50  \% Missing & Gain & 80\% Missing & Gain \\
      \hline
       \cite{cai2008framelet} & 32.77  & 1.87  & 29.50  & 1.58  & 26.75  & 6.33  & 26.03  & 1.51  \\
    \cite{cai2010simultaneous} & 33.07  & 1.57  & 29.75  & 1.33  & 30.98  & 2.10  & 26.23  & 1.31  \\
      \cite{cai2009split} & 32.64  & 2.00  & 29.25  & 1.83  & 26.60  & 6.48  & 25.89  & 1.65  \\
    \cite{elad2005simultaneous} & 33.11  & 1.53  & 29.88  & 1.20  & 30.65  & 2.43  & 25.85  & 1.69  \\
    \cite{li2013adaptive}  & 34.02  & 0.62  & 30.64  & 0.44  & 32.38  & 0.70  & 26.40  & 1.14  \\
     \cite{lim2013nonseparable} & 34.06  & 0.58  & 30.75  & 0.33  & 31.69  & 1.39  & 27.08  & 0.46  \\
    Alg. \ref{mainalg} & 34.64  &   & 31.08  &   & 33.08  &   & 27.54  &  \\
        \hline
    \end{tabular}
\end{table}%

It is observed that the PSNR values in Tables~\ref{tab:clean} and~\ref{tab:5122} obtained by our proposed Algorithm~\ref{mainalg} are generally higher than those obtained by other algorithms. The visual comparison of the inpainted images is shown in  Figures~\ref{Fig:3}--\ref{Fig:2}. For visual comparison,  we choose two typical zoomed-in portions of the original test images. One is a typical  texture-rich image  shown in Figure \ref{Fig:1}. The other is a typical cartoon image shown in Figure~\ref{Fig:2}. We see that our proposed Algorithm~\ref{mainalg} can recover both the texture part  and the cartoon part well.

\subsection{Numerical Experiments on Noisy Image Inpainting}
The inpainting algorithms in \cite{li2013adaptive,lim2013nonseparable} only consider the noiseless case. For the algorithms proposed in \cite{cai2008framelet,cai2010simultaneous,cai2009split,elad2005simultaneous}, the parameters of these iterative algorithms are often manually set and depend on the noise levels and test images.
As we already explained in Section~\ref{sec:introduction},
here we only report the performance of our proposed Algorithm~\ref{mainalg} for image inpainting with noise.
See Figure \ref{Fig:4} for an example.
We choose the i.i.d. zero-mean Gaussian noise with five different standard deviations $\sigma  = 5, 10, 20, 30, 50$ for $512\times 512$ test images in Figure~\ref{testimages}. The results are summarized in Table~\ref{tab:noisy512}.

\begin{table}[htbp]
  \centering
\caption{Performance in terms of PSNR values of our proposed image inpainting algorithm
under i.i.d. zero-mean Gaussian noise with noise standard deviation $\sigma=5, 10, 20, 30, 50$
for five $512\times 512$ test images in (f)-(j) of Figure~\ref{testimages}.
Top left table for inpainting
mask Text~3 in (m) of Figure~\ref{testimages}.
Top right table for inpainting
mask Text~4 in (n) of Figure~\ref{testimages}.
Bottom left table for
$50\%$ randomly missing pixels.
Bottom right table for
$80\%$ randomly missing pixels.
}
    \begin{tabular}{|c||c|c|c|c|c||c|c|c|c|c|}
   \hline
    $\sigma$  & Hill & Man & Boat & Barbara & Mandrill & Hill & Man & Boat & Barbara & Mandrill\\
    \hline
    $5$ & 33.48  & 33.39  & 32.80  & 34.04  & 29.06  & 31.23  & 30.87  & 29.80  & 31.31  & 26.96  \\
    $10$ & 31.45  & 31.43  & 31.04  & 31.81  & 27.37  & 29.90  & 29.64  & 28.77  & 29.84  & 25.92  \\
    $20$ & 29.18  & 29.04  & 28.84  & 28.98  & 25.02  & 28.14  & 27.85  & 27.29  & 27.70  & 24.20  \\
    $30$ & 27.84  & 27.58  & 27.41  & 27.17  & 23.53  & 27.00  & 26.66  & 26.20  & 26.24  & 22.98  \\
    $50$ & 26.18  & 25.77  & 25.55  & 24.89  & 21.74  & 25.56  & 25.11  & 24.75  & 24.28  & 21.42  \\
    \hline\hline
    $5$ & 32.58  & 32.45  & 32.51  & 33.41  & 25.91  & 28.87  & 28.47  & 27.98  & 27.69  & 22.04  \\
    $10$ & 30.70  & 30.64  & 30.65  & 31.10  & 25.12  & 27.90  & 27.55  & 27.08  & 26.66  & 21.81  \\
    $20$ & 28.41  & 28.28  & 28.20  & 27.99  & 23.64  & 26.39  & 26.06  & 25.54  & 24.67  & 21.26  \\
    $30$ & 27.04  & 26.82  & 26.64  & 25.93  & 22.38  & 25.33  & 24.92  & 24.42  & 23.30  & 20.65  \\
    $50$ & 25.36  & 25.00  & 24.71  & 23.56  & 20.75  & 23.95  & 23.39  & 22.90  & 21.85  & 19.75  \\
    \hline
    \end{tabular}%
\label{tab:noisy512}%
\end{table}%

\begin{figure}[htb]
\centering
\subfigure[Clean image]{\includegraphics[width=3.0cm, height=2.8cm]{7.png}}
\subfigure[Corrupted image]{\includegraphics[width=3.0cm, height=2.8cm]{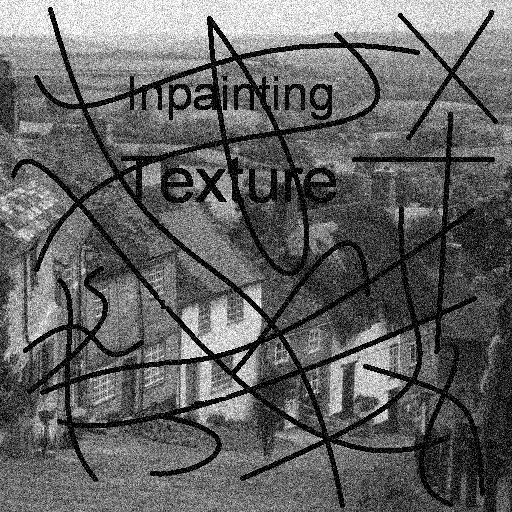}}
\subfigure[By Algorithm~\ref{mainalg}]{\includegraphics[width=3.0cm, height=2.8cm]{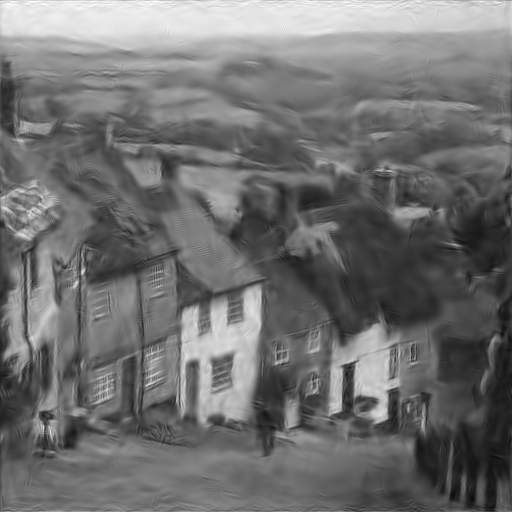}}
\caption{(a) is the $512\times 512$ clean image Hill.
(b) is the corrupted image with inpainting mask Text $3$ in (m) of Figure~\ref{testimages} and with i.i.d. zero-mean Gaussian noise having standard deviation $\sigma = 30$. (c) is the inpainted image by our proposed Algorithm~\ref{mainalg}.} \label{Fig:4}
\end{figure}

\begin{figure}[htb]
\centering
\subfigure[Clean image]
{\includegraphics[width=3.0cm, height=3cm]{9.png}}
\subfigure[Corrupted image]
{\includegraphics[width=3.0cm, height=3cm]{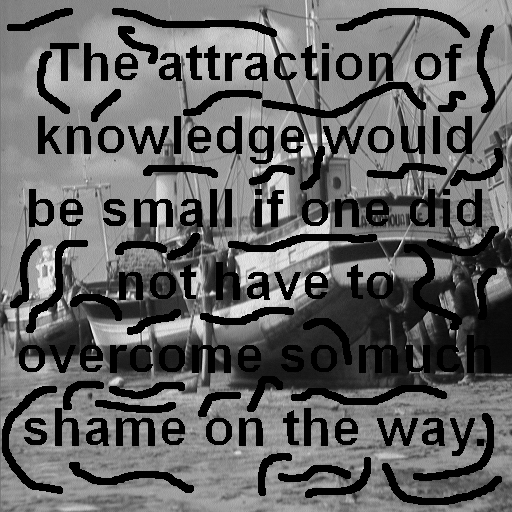}}
\subfigure[By  \cite{cai2008framelet}]
{\includegraphics[width=3.0cm, height=3cm]{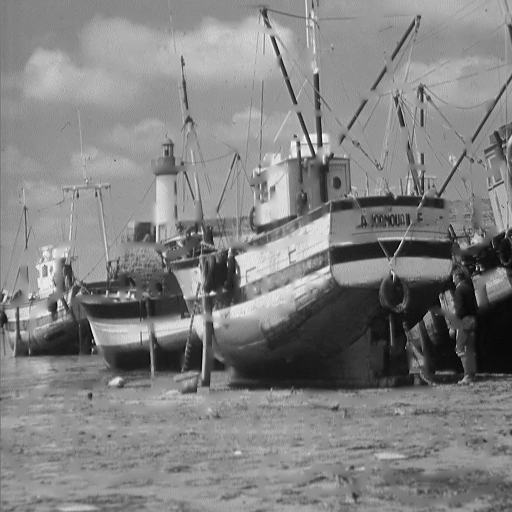}}
\subfigure[By  \cite{cai2010simultaneous}]
{\includegraphics[width=3.0cm, height=3cm]{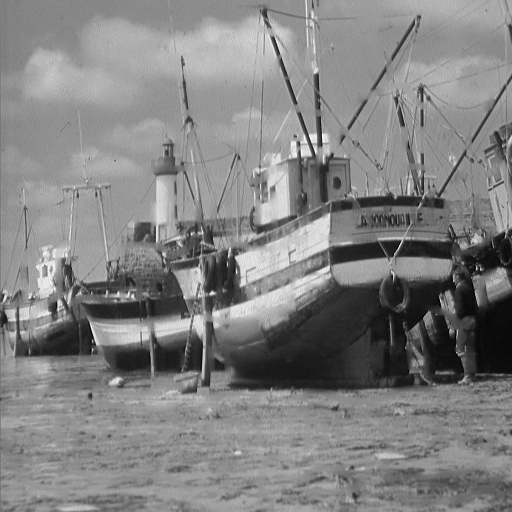}}
\subfigure[By \cite{cai2009split}]
{\includegraphics[width=3.0cm, height=3cm]{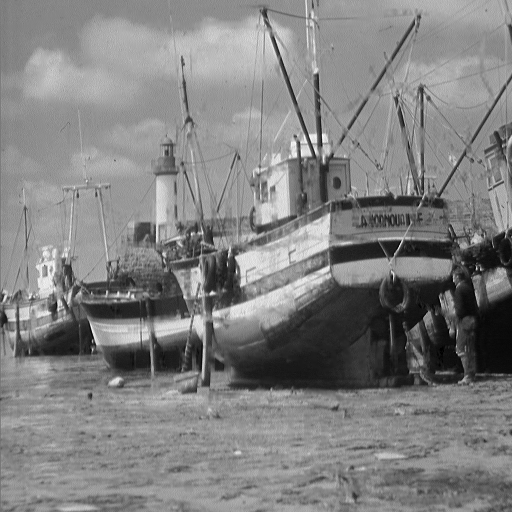}}\\
\subfigure[By  \cite{elad2005simultaneous}]
{\includegraphics[width=3.0cm, height=3cm]{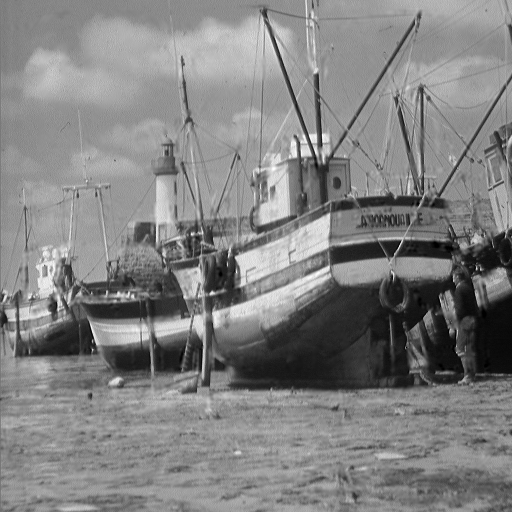}}
\subfigure[By  \cite{li2013adaptive}]
{\includegraphics[width=3.0cm, height=3cm]{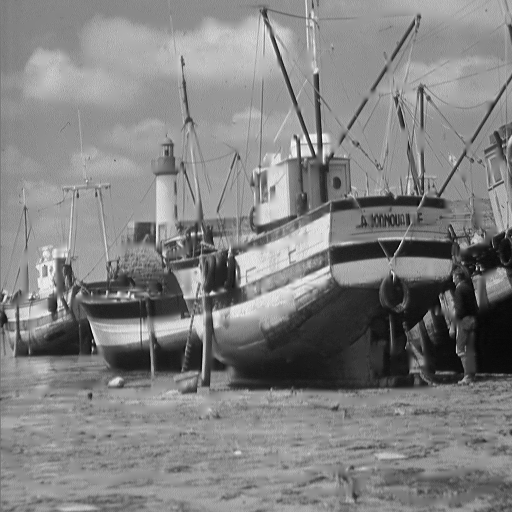}}
\subfigure[By  \cite{lim2013nonseparable}]
{\includegraphics[width=3.0cm, height=3cm]{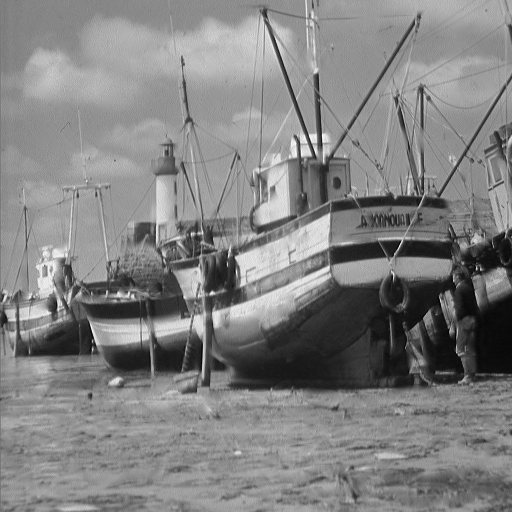}}
\subfigure[By  Algorithm \ref{mainalg}]
{\includegraphics[width=3.0cm, height=3cm]{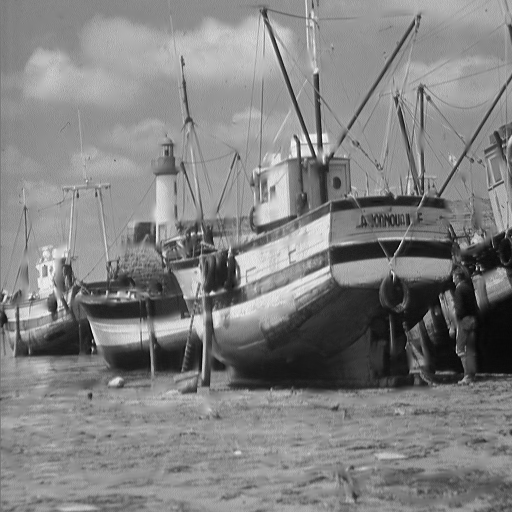}}
\caption{(a) is the $512\times 512$ clean image Boat. (b) is the corrupted image by Text~$4$ in (n) of Figure~\ref{testimages} without noise. (c)--(i) are the inpainted images by \cite{cai2008framelet,cai2010simultaneous,cai2009split,elad2005simultaneous,li2013adaptive,lim2013nonseparable} and our Algorithm~\ref{mainalg}.} \label{Fig:3}
\end{figure}

\begin{figure}[htb]
\centering
\subfigure[Clean image]
{\includegraphics[width=3.0cm, height=3cm]{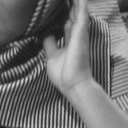}}
\subfigure[Corrupted image]
{\includegraphics[width=3.0cm, height=3cm]{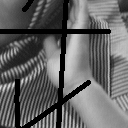}}
\subfigure[By  \cite{cai2008framelet}]
{\includegraphics[width=3.0cm, height=3cm]{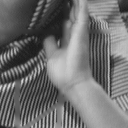}}
\subfigure[By  \cite{cai2010simultaneous}]
{\includegraphics[width=3.0cm, height=3cm]{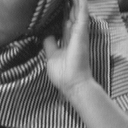}}
\subfigure[By  \cite{cai2009split}]
{\includegraphics[width=3.0cm, height=3cm]{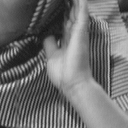}}\\
\subfigure[By  \cite{elad2005simultaneous}]
{\includegraphics[width=3.0cm, height=3cm]{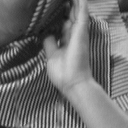}}
\subfigure[By  \cite{li2013adaptive}]
{\includegraphics[width=3.0cm, height=3cm]{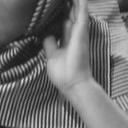}}
\subfigure[By  \cite{lim2013nonseparable}]
{\includegraphics[width=3.0cm, height=3cm]{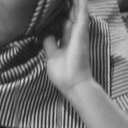}}
\subfigure[By  Algorithm \ref{mainalg}]
{\includegraphics[width=3.0cm, height=3cm]{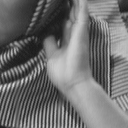}}
\caption{Zoomed-in portion of the $512\times 512$ clean image Barbara in (a), the corrupted image in (b) by inpainting mask Text~$3$ in (m) of Figure~\ref{testimages} without noise, and the inpainted images in (c)--(i) by \cite{cai2008framelet,cai2010simultaneous,cai2009split,elad2005simultaneous,li2013adaptive,lim2013nonseparable} and our Algorithm~\ref{mainalg}.} \label{Fig:1}
\end{figure}

\begin{figure}[hbt]
\centering
\subfigure[Clean image]
{\includegraphics[width=3.0cm, height=3cm]{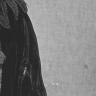}}
\subfigure[Corrupted image]
{\includegraphics[width=3.0cm, height=3cm]{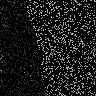}}
\subfigure[By  \cite{cai2008framelet} ]
{\includegraphics[width=3.0cm, height=3cm]{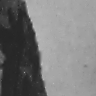}}
\subfigure[By  \cite{cai2010simultaneous}]
{\includegraphics[width=3.0cm, height=3cm]{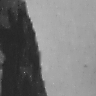}}
\subfigure[By  \cite{cai2009split}]
{\includegraphics[width=3.0cm, height=3cm]{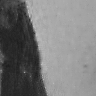}}\\
\subfigure[By  \cite{elad2005simultaneous}]
{\includegraphics[width=3.0cm, height=3cm]{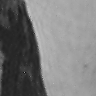}}
\subfigure[By  \cite{li2013adaptive}]
{\includegraphics[width=3.0cm, height=3cm]{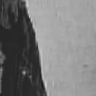}}
\subfigure[By  \cite{lim2013nonseparable}]
{\includegraphics[width=3.0cm, height=3cm]{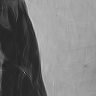}}
\subfigure[By  Algorithm \ref{mainalg}]
{\includegraphics[width=3.0cm, height=3cm]{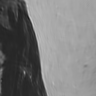}}
\caption{
Zoomed-in portion of the $512\times 512$ clean image Man in (a), the corrupted image in (b) by $50\%$ randomly missing pixels without noise, and the inpainted images in (c)--(i) by \cite{cai2008framelet,cai2010simultaneous,cai2009split,elad2005simultaneous,li2013adaptive,lim2013nonseparable} and our Algorithm~\ref{mainalg}.} \label{Fig:2}
\end{figure}

\section{Conclusions}
Motivated by the balanced approach in
\cite{cai2008framelet,cai2009split}
and the directional tensor product complex tight framelets in \cite{HanMMNP2013,HanZhaoSIIMS2014},
we proposed in this paper a frame-based image inpainting algorithm. Numerical results show that our proposed inpainting algorithm can restore corrupted images with better quality than those recovered by the state-of-the-art frame-based iterative inpainting algorithms in
\cite{cai2008framelet,cai2010simultaneous,cai2009split,elad2005simultaneous,li2013adaptive,lim2013nonseparable}.
Moreover, our proposed algorithm performs well for image inpainting under noise.
Motivated by the grouping effect property of elastic net in statistics,
we proved in Theorem~\ref{thm:ge}
the grouping effect property for the balanced approach using redundant frame systems. The proved grouping effect property in Theorem~\ref{thm:ge} partially explains the effectiveness of the balanced approach in image processing.
As a future work, we expect that our results on image inpainting could be further improved by exploring the freedom in the design of directional tensor product complex tight framelets in \cite{HanMMNP2013,HanZhaoSIIMS2014} or by using directional nonseparable tight framelets in \cite{HanACHA2012}.

\noindent \emph{Acknowledgement:}
The authors would like to thank Jian-Feng Cai and Yan-Ran Li for providing us their source matlab codes for their inpainting algorithms in \cite{cai2008framelet,cai2010simultaneous,cai2009split,li2013adaptive}.
The authors also thank Zhenpeng Zhao for providing us his source matlab code implementing the discrete framelet transform using $\tpctf_6$.

\goodbreak

\end{document}